\documentclass[12pt]{article}
\usepackage{amsmath,amssymb,epsf,epsfig,amsthm,times,ifthen,epic,psfrag,bbm, mathabx}
\usepackage{pgfplots}
\usepackage{enumerate}
\usepackage{fancyhdr}
\usepackage{setspace} 
\usepackage{lastpage}
\usepackage{tikz}
\usepackage[displaymath, mathlines]{lineno}
\usetikzlibrary{patterns}

\def\submitteddate{March 20, 2026}

\renewcommand{\baselinestretch}{1}
\usepackage{hyperref}
\hypersetup{
    colorlinks=true,
    citecolor=black,
    filecolor=black,
    linktoc=all,     %set to all if you want both sections and subsections linked
    linkcolor=blue,
    pdfpagemode=UseNone, % This prevents the bookmarks sidebar from popping up.
    urlcolor=black
}

\begin{document}

\newcommand{\creationtime}{\today \ \ @ \theampmtime}

\pagestyle{fancy}
\renewcommand{\headrulewidth}{0cm}
%\chead{\footnotesize{Super-regular matrices and MDS codes}}
%\rhead{\footnotesize{\reviseddate}}
%\rhead{\footnotesize{\creationtime}}
\rhead{\footnotesize{\submitteddate}}
%\lhead{\footnotesize{Appuswamy, Bazzani, Congero, Connelly, Ekaireb, Zeger}}
\lhead{\footnotesize{\textit{Probability of super-regular matrices and MDS codes over finite fields}}}
\cfoot{Page \arabic{page} of \pageref{LastPage}} % Need to include package ``lastpage''

\lfoot{\footnotesize{
Sections: 
\ref{sec:introduction}
%\ref{sec:background}
\ref{sec:MDS-codes}
\ref{sec:Contiguous-Super-Regular-Square-Matrices}
\ref{sec:small-matrices}
%\ref{sec:3x3}
%\ref{sec:4x4}
\ref{sec:Conjecture}
\ref{sec:Appendix}
\nameref{references}
}}

\renewcommand{\qedsymbol}{$\blacksquare$} % Used to control qed generated by \end{proof}

%\numberwithin{equation}{section}    % Make sure this line comes before the next line redefining \theequation
%\renewcommand{\theequation}{\arabic{section}.\arabic{equation}}

\newtheorem{theorem}              {Theorem}     [section]
\newtheorem{lemma}      [theorem] {Lemma}
\newtheorem{corollary}  [theorem] {Corollary}
\newtheorem{proposition}[theorem] {Proposition}
\newtheorem{conjecture} [theorem] {Conjecture}

\theoremstyle{remark}
\newtheorem{remark}     [theorem] {Remark}
\newtheorem{algorithm}  [theorem] {Algorithm}

\theoremstyle{definition}         
\newtheorem{definition} [theorem] {Definition}
\newtheorem{example}    [theorem] {Example}
\newtheorem*{claim}  {Claim}
\newtheorem*{notation}  {Notation}
%%%%%%%%%%%%%%%%%%%%%%%%%%%%%%%%%%%%%%%%%%%%%%%%%%
\renewcommand{\complement}{c}
\newcommand{\transpose}{t}
\newcommand{\E}{\mathbb{E}}
\newcommand{\indicator}[1]{\mathbbm{1}_{#1}}

\newcommand{\minij}{\text{min}(i,j)}

\newcommand{\field}[1]{\mathbb{F}_{#1}}
\newcommand{\fq}{\field{q}}
\newcommand{\R}{\mathbb{R}}
\newcommand{\Z}{\mathbb{Z}}
\newcommand{\Q}{\mathbb{Q}}
\newcommand{\N}{\mathbb{N}}
\newcommand{\floor}[1]{\left\lfloor #1 \right\rfloor}
\newcommand{\ceil}[1]{\left\lceil #1 \right\rceil}

\newcommand{\rank}[1]{\mbox{\textit{rank}}\!\left(#1\right)}

\newcommand{\Var}{\operatorname{Var}}
\newcommand{\Cov}{\operatorname{Cov}}

\newcommand{\Hide}[1]{}

\DeclareRobustCommand{\stirling}{\genfrac\{\}{0pt}{}}

\newcommand{\dsum}{\displaystyle\sum}

% Change the symbol \phi in the next 6 definitions...
\newcommand{\CountMDS} [3]{C^{\text{MDS}}_{#2,#3}(#1)} % q,n,k Number of [n,k] MDS codes over GF(q). \gamma
\newcommand{\CountSR}  [2]{C^{\text{SR}}_{#2}(#1)}  % q,k Number of kxk super-regular GF(q) matrices 
\newcommand{\CountCSR} [2]{C^{\text{CSR}}_{#2}(#1)} % q,k Number of contiguous kxk super-regular GF(q) matrices
\newcommand{\ProbMDS}  [3]{P^{\text{MDS}}_{#2,#3}(#1)} % q,n,k Prob [n,k] code over GF(q) is MDS 
\newcommand{\ProbSR}   [2]{P^{\text{SR}}_{#2}(#1)}  % q,k Prob kxk GF(q) matrix is super-regular 
\newcommand{\ProbCSR}  [2]{P^{\text{CSR}}_{#2}(#1)} % q,k Prob kxk GF(q) matrix is contiguous super-regular 

\newcommand{\UltraProb}[1]{p_{\text{SR}}(#1)}
\newcommand{\StrongProb}[1]{p_{\text{CSR}}(#1)}
\newcommand{\Inflate}{1.0}
\newcommand{\To}{\to}

%%%%%%%%%%%%%%%%%%%%%%%%%%%%%%%%%%%%%

\begin{titlepage}

\setcounter{page}{1}

\title{Probability of super-regular matrices\\
        and MDS codes over finite fields%
\thanks{
\indent 
R. Appuswamy is with IBM Research, San Jose, CA;
M. Bazzani is with the Department of Mathematics,
University of California, San Diego, La Jolla, CA;
S. Congero is in San Diego, CA;
J. Connelly is with Seagate Technology, Minneapolis, MN;
M. Ekaireb is with Canyon Crest Academy,
San Diego, CA;
K. Zeger is with the Department of Electrical and Computer Engineering, 
University of California, San Diego, La Jolla, CA. \ \ 
This work used 
% [resource-name] 
CPU time on the Expanse supercomputer at 
% [resource provider] 
the San Diego Supercomputer Center
through allocation 
%[allocation number] 
ELE250044
from the Advanced Cyberinfrastructure
Coordination Ecosystem: Services \& Support (ACCESS) program, which is
supported by U.S. National Science Foundation grants 
\#2138259, \#2138286, \#2138307, \#2137603, and \#2138296.
% (rathnam@ucsd.edu, \ zeger@ucsd.edu)
}}

\author{
Rathinakumar Appuswamy \and 
Marco Bazzani \and 
Spencer Congero \and 
Joseph Connelly \and 
Matthew Ekaireb \and 
Kenneth Zeger}

% \date{
% \textit{IEEE Transactions on Information Theory\\
% %Version: \version\\
% %Created: \creationtime \\
% Submitted: \submitteddate
% %Revised:  \reviseddate 
% %Revised Draft:   \creationtime 
% % \today \\
% % \large{Draft :  *** Do Not Distribute ***} 
% }}

\maketitle

\begin{abstract}
Let $C$ be an $[n,k]$ linear code chosen uniformly at random 
over a finite field $\fq$ of size $q$.
The following asymptotic probability
of $C$ being maximum distance separable (MDS)
as $q,n,k\To\infty$ is known:
If $\frac{1}{q}\binom{n}{k} \To 0$, 
then $P(C\ \text{is MDS}) \To 1$.
We demonstrate that this growth rate is in fact a threshold by proving:
If $\frac{1}{q}\binom{n}{k} \To \infty$, 
then $P(C\ \text{is MDS}) \To 0$.
A matrix is (\textit{contiguous}) \textit{super-regular} 
if all of its (contiguous) square submatrices are nonsingular.
The above results imply that 
for any
$k \times k$ matrix $A$ chosen uniformly at random over $\fq$,
the following hold:
If $\frac{4^k/\sqrt{k}}{q} \To 0$, 
then $P(A \text{ is super-regular}) \To 1$.
If $\frac{4^k/\sqrt{k}}{q}\To \infty$, 
then $P(A \text{ is super-regular}) \To 0$.
We also obtain the following asymptotic probabilities
for two variations of the above questions:
If $\frac{1}{q}\binom{n}{k} \To \lambda \in (0,\infty)$
and $k/n\To 0$,
then $P(C\ \text{is MDS}) \To e^{-\lambda}$.
If $\frac{k^3/3}{q} \To \lambda \in [0,\infty]$,
then $P(A \text{ is contiguous super-regular}) \To e^{-\lambda}$.
The number of super-regular  $3\times 3$ matrices is known to be a polynomial in $q$.
We show that the number of contiguous super-regular $3\times 3$ matrices 
is also a polynomial.
Finally, for $4\times 4$ matrices,
we show that
the number of super-regular matrices is not a polynomial,
nor even a quasi-polynomial of period less than $7$,
whereas our experimental evidence suggests that
the number of contiguous super-regular matrices is a polynomial.
\end{abstract}

\thispagestyle{empty}
\end{titlepage}

\tableofcontents
%\linenumbers

\clearpage

\section{Introduction}
\label{sec:introduction}

Random matrices are a widely studied topic of 
linear algebra and probability theory and have numerous 
applications in engineering and physics.
There are two main types of random matrices 
that have been studied,
namely those over infinite fields and those over finite fields.

For matrices with 
continuously distributed elements over an infinite field such as
the real or complex numbers,
prior studies have typically focused on the eigenvalues
as the matrices become asymptotically large
(e.g., see the books
\cite{Akemann-Baik-DiFrancesco-book-2011},
\cite{Anderson-Guionnet-Zeitouni-book-2010},
\cite{Mehta-book-2004},
\cite{Potters-Bouchaud-book-2020}).
If a matrix's elements are iid continuous random variables,
then its determinant is almost surely nonzero,
since the determinant is a polynomial function of the elements.
That is, singular matrices are unlikely in such cases.
Other studies 
dating back to the 1960s
have considered the invertibility of matrices whose
elements are real or complex discrete random variables.
(e.g., 
\cite{Komlos-1968},
\cite{Kozlov-1966}
).

Our work focuses on
matrices whose elements are chosen randomly 
from a finite field.

\subsection{Random matrices over finite fields}
\label{sec:matrices-finite-fields}

Let $\fq$ denote a finite field of size $q$.
A matrix is said to be \textit{random over $\fq$}
if its entries are chosen independently and uniformly 
at random from the elements of $\fq$.

Random matrices over finite fields have also been
an area of extensive study
(e.g.,
\cite{Cooper},
\cite{Fine-Herstein-1958},
\cite{Fulman-2000},
\cite{Fulman-2002},
\cite{Hansen-Schmutz-1993},
\cite{Kahn-Komlos-2001},
\cite{Maples-preprint},
\cite{Salmond-Grant-Grivell-Chan-2016},
\cite{Sanna-2024},
\cite{Schmutz-1995}).
The probability that a random $k \times k$ matrix over $\fq$
is nonsingular is 
\begin{align*}
\pi (q,k) 
&= \displaystyle\prod_{i=0}^{k-1} \frac{q^k - q^i}{q^k}
= \prod_{i=1}^{k} \left( 1 - q^{-i} \right)
\end{align*}
which is well known
(e.g., \cite{Cooper}). % pp. 1-2
This is because if the first $i$ rows of a $k\times k$ matrix over $\fq$ 
are linearly independent, 
then their span has size $q^i$,
so the number of choices for the $(i+1)$th row 
to retain independence is $q^k - q^i$.
It is easy to see that $\pi (q,k)$
is monotone increasing in $q$ and decreasing in $k$.
It is known that 
$\displaystyle\lim_{k \To \infty} \pi (q,k) \in (0,1)$
for any fixed $q$.
For example, 
$\displaystyle\lim_{k \To \infty} \pi (2,k) \approx 0.288$.
It is also known 
(e.g., Lemma~\ref{lem:nonsingular-prob}) that
$\displaystyle\lim_{q \To \infty} \pi (q,k) = 1$
for any fixed $k$,
and
$\displaystyle\lim_{q,k \To \infty} \pi (q,k) = 1$,
so the probability that a random $k\times k$ matrix over $\fq$
is nonsingular is asymptotically one,
no matter how
the field size $q$ and the matrix dimension $k$ both tend to infinity.

A \textit{submatrix} of a given matrix is a matrix formed by
deleting possibly some but not all of its rows and
possibly some but not all of its columns.
A \textit{contiguous submatrix} 
further requires that the non-deleted rows and columns each be consecutive%
\footnote{We will consistently use the terminology ``non-contiguous'' to mean 
not necessarily contiguous.}.

An alternative to determining if a random matrix is nonsingular
is to ask if certain collections of submatrices are nonsingular.
As an example,
the probability that all $1\times 1$ 
submatrices of a $k\times k$ random matrix over $\fq$ are nonsingular
(i.e., each entry is non-zero) is 
$(1 - q^{-1})^{k^2}$,
which tends to $e^{-\lambda}$ 
if 
$k^2/q\To \lambda\in [0,\infty]$
as $q$ and $k$ tend to infinity.
By assumption,
the $1\times 1$ submatrices are independent,
making the calculation of their collective nonsingularity easy.

In contrast,
our study considers collections containing overlapping submatrices,
which are not independent,
making analysis more difficult.
A matrix 
is called \textit{super-regular} if all of its square submatrices are nonsingular,
and 
is called \textit{contiguous super-regular} if all of its contiguous square submatrices are nonsingular.
Super-regular matrices have been studied in various contexts
(e.g.,
\cite{Almeida-Napp-2020}, 
\cite{Almeida-Napp-Pinto-2013}, 
\cite{Gluesing-Rosenthal-Smarandache-2006}, % Def 3.3 on pg. 588
\cite{Lovett-MDS-conjecture},
\cite{Maturana-Rashmi-2012},
\cite{Roth-book}, % p. 335
\cite{Roth-Lempel-1989}, % 1st paragraph of pg. 1314.
\cite{Roth-Seroussi-1985},
\cite{Yildiz-Hassibi-MDS-conjecture}).
One known class of $k\times k$ super-regular matrices 
are the \textit{Cauchy matrices} 
\cite[p. 168]{Roth-book},
whose $(i,j)^{th}$ elements are $(x_i - y_j)^{-1}$,
where the $2k$ quantities
$x_1, \dots, x_k, y_1, \dots, y_k$ are
distinct elements of a field $\fq$.
Thus, super-regular matrices 
exist for all $k$,
whenever $q\ge 2k \ge 4$.

% --------------------------------------------------------------------------
One goal of our work is
to analyze the asymptotic probability 
that a random $k\times k$ matrix over $\fq$
is super-regular
(or alternatively, contiguous super-regular).
In particular, 
in this paper we examine conditions
on how fast the field size $q$ should grow toward infinity 
relative to $k$
in order to guarantee 
(with high probability) 
whether or not
all submatrices of a $k\times k$ matrix are nonsingular.
We also consider the less restrictive case 
where nonsingularity is only required for contiguous square submatrices.

Lemma~\ref{lem:fullColRankProb}
notes that the probability that a random square matrix of any size over $\fq$ is singular
is close to $1/q$ as $q$ grows.
If $S$ is a set of independent square submatrices of a given random matrix over $\fq$,
then the probability of all such submatrices being nonsingular 
is approximately
$\left( 1 - \frac{1}{q}\right)^{|S|} \sim e^{-\lambda}$
as $\frac{|S|}{q} \To \lambda$.
It turns out
(Theorem~\ref{thm:main-contiguous})
for contiguous super-regular square matrices,
even with dependent contiguous submatrices,
this asymptotic probability is in fact correct.
In the non-contiguous case,
we show 
(Corollary~\ref{cor:main-superregular})
this behavior holds for the cases of $\lambda=0,\infty$
and give experimental evidence 
(Section~\ref{sec:Conjecture})
that it is plausible for more general $\lambda$.

Overall,
for both the non-contiguous and contiguous submatrix cases,
we obtain limiting probabilities of random square matrices 
over a finite field being super-regular of the form
\begin{align*}
P(\text{super-regular}) \To
e^{- \frac{\text{number of nonsingular submatrix constraints}}
          {\text{field size}}} 
\end{align*}
as the field size and matrix size tend to infinity.
In the non-contiguous case,
the number of submatrices constrained to be nonsingular is 
$\displaystyle\sum_{i=1}^k\binom{k}{i}^2 = \binom{2k}{k} - 1 
 \sim \frac{4^k}{\sqrt{k}}$
and in the contiguous case
the number of submatrices constrained to be nonsingular is 
$\displaystyle\sum_{i=1}^k (k-i+1)^2 = \sum_{i=1}^k i^2
 \sim \frac{1}{3}k^3$.
The results reveal a 
probability thresholding relationship in growth rates between
the matrix dimensions and the field size.

\subsection{MDS Codes}

The probability that a random square matrix is super-regular
is closely related to a more general question from coding theory.

%-----------------------

A $q$-ary $[n,k]$ \textit{linear code} 
is a subspace of $\fq^n$ of size $q^k$.
A \textit{generator matrix} for a $q$-ary $[n,k]$ linear code
is a $k\times n$ matrix over $\fq$ 
whose row span is the code.
Thus, a generator matrix must have full row rank $k$
in order for
its span to have size $q^k$.
A \textit{parity check matrix} for a $q$-ary $[n,k]$ linear code
is an $(n-k)\times n$ matrix over $\fq$ whose
null space is the code.
A generator matrix is \textit{systematic} if it can be written with a $k\times k$
identity matrix on the left as $[I_k | A]$ 
and the corresponding systematic parity check matrix 
is $[-A^t | I_{n-k}]$,
with an $(n-k)\times (n-k)$ identity matrix on the right.
In practice, a length-$k$ input vector of $q$-ary 
information symbols left-multiplies
the generator matrix to produce a length-$n$ output vector of $q$-ary 
channel symbols for transmission
across a noisy channel,
and a parity check matrix multiplies a received length-$n$ vector of $q$-ary symbols
from a channel in order to estimate the $k$ originally transmitted information symbols.

The \textit{minimum distance} $d$ of a code is the 
smallest number of corresponding components 
that differ between any two distinct elements of a code.
Not all values of $n,k,q,d$ are possible for linear codes.
The existence (or not) of linear codes for various instances of these values has been
a topic of interest for many decades.
The ``Singleton bound'' for linear codes
is due to 
Komamiya~\cite{Komamiya-1953} in 1953, 
Joshi~\cite{Joshi-1958} in 1958, and 
Singleton~\cite{Singleton-1964} in 1964,
and upper bounds the minimum distance $d$ as
$d \le n-k+1$
for all linear $[n,k]$ codes over a field of size $q$.

This bound is of particular interest when it is met with equality.
A $q$-ary $[n,k]$ linear code is \textit{maximum distance separable} (MDS)
if its minimum distance is $d=n-k+1$.
We say that a generator matrix (or a parity check matrix) \textit{is MDS}
if its corresponding linear code is MDS.
Some connections between an MDS code
and the nonsingularity of various submatrices 
of its generator or parity check matrix are known,
and are summarized in the next lemma,
which is a combination of
Corollary 3 % p. 319
and
Theorem 8   % p. 321
in Chapter 11
of the MacWilliams-Sloane textbook~\cite{MacWilliams-Sloane-book}.

\begin{lemma}
Let $C$ be an $[n,k]$ linear code over $\fq$ 
with generator matrix $G$, 
systematic generator matrix $[I|A]$,
and
parity check matrix $H$.
Then the following are equivalent:
\begin{itemize}
\setlength\itemsep{-0.15cm}
\item [(a)] $C$ is MDS.
%           G has dimensions k x n
\item [(b)] Any $k$ columns of $G$ are linearly independent.
\item [(c)] $A$ is super-regular.
%           H has dimensions (n-k) x n
\item [(d)] Any $n-k$ columns of $H$ are linearly independent.
\end{itemize}

\label{lem:MDS-conditions}
\end{lemma}

% ---------------------------------------------------------------------------
% MDS Conjecture
The \textit{MDS Conjecture} 
(e.g., see Hirschfeld~\cite[Conjecture 5.1]{Hirschfeld-2003}) % p. 8
originated in
1955 by Segre~\cite{Segre-1955}
and states that any $[n,k]$ MDS code over $\fq$ satisfying $k\le q$
must also have $n\le q+1$,
unless $q$ is even and $k\in\{3,q-1\}$,
in which case $n\le q+2$.
As an example,
$[n,k]$ Reed-Solomon codes over $\fq$ are MDS
whenever $n=q-1$
(e.g., \cite{MacWilliams-Sloane-book}). % p. 317

Instead of focusing on the existence of MDS codes for various values of
the parameters $n,k,q$,
we examine the prevalence of such codes,
namely their probability of occurring when linear codes are chosen randomly.
Our results 
(e.g., Theorem~\ref{thm:main-MDS}(a))
demonstrate that the frequency of occurrence of MDS codes tends to zero when the
field size $q$ grows slowly compared to $\binom{n}{k}$.

We note that another conjecture known as the 
\textit{GM-MDS Conjecture}
upper bounds $n$ for a given $q$ and $k$ 
and a prescribed set of zero entries for MDS generator matrices
\cite{Dau-Song-Dong-Yuen-2013},
\cite{Dau-Song-Dong-Yuen-2014}.
This conjecture does not directly relate to our results.

% ---------------------------------------------------------------------------

%
By a \textit{random $q$-ary linear $[n,k]$ code} 
we mean an $[n,k]$ linear code chosen uniformly at random from
among all $q$-ary linear $[n,k]$ codes.

We will examine the asymptotic probability that a 
randomly chosen linear code is MDS 
as $n$, $k$, and $q$ all tend to infinity,
and show that there is an asymptotic relationship 
between these three parameters
that acts as a threshold for determining whether 
such probability goes to zero or one.

For our asymptotic results,
Lemma~\ref{lem:multiple-defs-of-code} 
allows the use of alternative conditions on MDS codes.
For example,
one can use a
\textit{random $k\times n$ generator matrix from $\fq$ in systematic form}, 
that is,
choose the elements of
the $k\times (n-k)$ non-identity portion of 
the generator matrix iid uniformly from $\fq$.
Alternatively, one can choose a 
\textit{random $k\times n$ matrix from $\fq$}, 
that is, 
choose all $kn$ elements iid uniformly from $\fq$.
In this case, however,
not all choices yield a generator matrix since the matrix may not have full row rank.
However, the probability of obtaining such a matrix with row rank at most $k-1$
diminishes to zero as $q\To\infty$ 
so it becomes a negligible portion of the chosen  matrices
(also by Lemma~\ref{lem:multiple-defs-of-code}).

We demonstrate bounds on the probability that a random linear $[n,k]$
code is MDS and then specialize the results to the case where
$n=2k$ to obtain asymptotic probabilities for 
random $k\times k$ matrices to be super-regular.

%-----------------------

Throughout this paper,
let $\CountSR{q}{k}$ 
(respectively,  $\CountCSR{q}{k}$)
be the number of $k\times k$ matrices over $\fq$ that are
super-regular 
(respectively, contiguous super-regular),
and let $\ProbSR{q}{k} = \CountSR{q}{k} / q^{k^2}$ 
(respectively,  $\ProbCSR{q}{k}  = \CountCSR{q}{k} / q^{k^2}$)
be the probability that a random $k\times k$ matrix over $\fq$ is
super-regular 
(respectively, contiguous super-regular).
Also, let $\CountMDS{q}{n}{k}$ be the number of linear $[n,k]$ 
MDS generator matrices over $\fq$,
and let $\ProbMDS{q}{n}{k} = \CountMDS{q}{n}{k} / q^{k(n-k)}$ 
be the probability that a random linear $[n,k]$ generator matrix over $\fq$ is MDS.

Some asymptotic properties of the probability $\ProbMDS{q}{n}{k}$ 
are implied from the literature.
Ghorpade and Lachaud~\cite[Lemma 5.1]{Ghorpade-Lachaud-2001}
gave a family of upper and lower bounds 
$\mathcal{B}_{2r+1} \le \CountMDS{q}{n}{k} \le \mathcal{B}_{2r}$
for $r=0,1,2,\dots$.
An asymptotic growth rate follows from this for fixed $k$ and $n$,
as $q$ grows,
namely
\begin{align*}
\ProbMDS{q}{n}{k} 
&= 1 - \left( \binom{n}{k} - 1 \right) q^{-1} + O(q^{-2}).
\end{align*}

Kaipa~\cite{Kaipa-2013} improved the accuracy of these bounds,
again in the case of fixed $k$ and $n$ while $q$ grows.
There do not appear to be similar asymptotic MDS enumeration estimates 
as all three quantities $n,k,q$ grow,
which is the focus of our work.

For some values of $n$ and $k$ the function $\CountMDS{q}{n}{k}$
is known exactly.
For example, 
$\CountMDS{q}{n}{k}$ is known
for all $n$ if $k=1,2$,
and for all $n\le 9$ if $k=3$
(e.g., see the references in \cite{Kaipa-2013}).
One currently unknown case is $\CountMDS{q}{8}{4}$. 
This value is the same as $\CountSR{q}{4}$,
that is, the number of super-regular $4\times 4$ matrices over $\fq$.
In fact, it has not been previously reported in the literature if $\CountMDS{q}{8}{4}$
is a polynomial or even a quasi-polynomial function of $q$.

Each linear $[n,k]$ MDS code over $\fq$ corresponds to exactly one unique
systematic $k\times n$ generator matrix over $\fq$.
The total number of $k\times (n-k)$ matrices over $\fq$ whose left-most $k\times k$
submatrix is the identity is
$q^{k(n-k)}$.
Thus the probability that a random $k\times n$ matrix in systematic form
is an MDS generator matrix is upper and lower bounded 
via~\cite{Ghorpade-Lachaud-2001}  as 
$\mathcal{B}_{2r+1} q^{-k(n-k)} \le 
 \ProbMDS{q}{n}{k} \le 
 \mathcal{B}_{2r} q^{-k(n-k)}$.

If we take $n=2k$,
then $\CountMDS{q}{2k}{k} = \CountSR{q}{k}$
by Lemma~\ref{lem:MDS-conditions},
so the bounds become
$\mathcal{B}_{2r+1} q^{-k^2} \le \ProbSR{q}{k} \le \mathcal{B}_{2r} q^{-k^2}$.
It can be shown in this case that 
$\mathcal{B}_1q^{-k^2}\To 1$ 
(and thus $\ProbSR{q}{k}\To 1$)
whenever 
both $k,q\To\infty$ and
$\frac{4^k/\sqrt{k}}{q} \To 0$.
On the other hand,
it is not known if this rate is tight.
We investigate what happens to the probability $\ProbSR{q}{k}$
when $\frac{4^k/\sqrt{k}}{q} \To \lambda \in (0,\infty]$.
It does not appear straightforward 
to answer this question
using \cite{Ghorpade-Lachaud-2001}.
In fact, direct calculation shows that neither of the upper bounds
$\mathcal{B}_0q^{-k^2}$ 
nor $\mathcal{B}_2 q^{-k^2}$ 
generally tends to $0$ 
when $\frac{4^k/\sqrt{k}}{q}\To \infty$.
For each fixed $k$ and $q$,
the upper and lower bounds in \cite{Ghorpade-Lachaud-2001}
eventually coincide when $r>\binom{2k}{k}$,
but their complexity of calculation appears unwieldy,
even for $r=2$.

\subsection{Our Contributions}

Our results on super-regular matrices are of two types: 
asymptotic probabilities with
thresholds or exact limits
(Sections~\ref{sec:MDS-codes} 
and \ref{sec:Contiguous-Super-Regular-Square-Matrices}),
and enumerations for small cases,
i.e., $3\times 3$ and $4\times 4$ matrices
(Section~\ref{sec:small-matrices}).

In Sections~\ref{sec:MDS-codes}
and~\ref{sec:Contiguous-Super-Regular-Square-Matrices}
we focus on the asymptotic behavior
of the probability that a random $k \times k$ matrix over $\fq$
is super-regular (or contiguous super-regular)
as both the matrix and field sizes grow.
We show
in Theorem~\ref{thm:main-MDS}
that if $C$ is a linear $[n,k]$ code chosen uniformly at random
over a field of size $q$,
where $n,q \To\infty$,
then the following three results hold:
if $\frac{1}{q}\binom{n}{k}\To 0$, then $P(C \text{ is MDS}) \To 1$;
if $\frac{1}{q}\binom{n}{k}\To \infty$, then $P(C \text{ is MDS}) \To 0$;
if $\frac{1}{q}\binom{n}{k}\To \lambda\in (0,\infty)$ and
      $k/n \To 0$, then $P(C \text{ is MDS}) \To e^{-\lambda}$.

Theorem~\ref{thm:main-MDS} shows that there is an asymptotic threshold at
$q \sim \binom{n}{k}$.
If $q$ asymptotically dominates this threshold
then almost all such matrices are MDS,
and
if $q$ is asymptotically dominated by the threshold
then vanishingly few $k \times n$ matrices over $\fq$ are MDS.
With the extra restriction that
$k=o(n)$, we can also characterize the intermediate cases where
$\binom{n}{k}/q \To \lambda \in (0, \infty)$.

These results then imply
in Corollary~\ref{cor:main-superregular}
that if $A$ is a random $k \times k$ matrix over a field of size $q$,
where $q, k \To \infty$,
then the following two results hold:
if $\frac{4^k/\sqrt{k}}{q} \To 0$,
 then $P(A \text{\ is super-regular}) \To 1$;
if $\frac{4^k/\sqrt{k}}{q} \To \infty$,
 then $P(A \text{\ is super-regular}) \To 0$.
We provide an elementary proof of the first of these two results
as an alternative to that deducible
from \cite{Ghorpade-Lachaud-2001}.

We also consider the asymptotics when only the contiguous submatrices
of a random matrix are required to be nonsingular.
Our main result
in Theorem~\ref{thm:main-contiguous}
is that
$\ProbCSR{q}{k}\To e^{-\lambda}$
for any $\lambda\in [0,\infty]$
whenever
$k,q\To\infty$ and
$\frac{\frac{1}{3}k^3}{q} \To \lambda$.
Similar to the (non-contiguous) super-regular case,
this result shows that
an asymptotic probability threshold holds
for contiguous super-regular matrices as well,
but with the growth rate $q \sim \frac{1}{3} k^3$.

The asymptotic threshold required in Theorem~\ref{thm:main-MDS}
(and Corollary~\ref{cor:main-superregular})
is $\frac{1}{q}\binom{n}{k}$
and the threshold required in
Theorem~\ref{thm:main-contiguous}
is $\frac{k^3/3}{q}$.
These thresholds have the common feature that both count the number of
submatrices of the original random matrix that must be nonsingular.
That is, in both cases,
the asymptotic probability threshold for a matrix to be super-regular
requires the field size to grow proportionally
to the number of submatrices that are constrained.

% ---------------------------------------------------------------------------

In Sections~\ref{sec:3x3} and~\ref{sec:4x4}
we consider nonasymptotic results for small matrices.
The numbers of super-regular
and contiguous super-regular
$1\times 1$ and $2\times 2$ matrices over $\fq$ are polynomial functions of $q$,
and the number of $3 \times 3$ super-regular matrices is known
(e.g., \cite[Proposition 3.2]{Skorobogatov-1992})
to be $\CountSR{q}{3} = (q-1)^5 (q-2)(q-3)(q^2-9q+21)$.
We show
in Theorem~\ref{thm:3x3}
that the number of $3 \times 3$ contiguous super-regular matrices over $\fq$
is $\CountCSR{q}{3} = (q-1)^5 (q-2)(q-3)(q^2-4q+5)$.
Therefore, $\CountCSR{q}{3} > \CountSR{q}{3}$
whenever $q \ge 4$,
since $q^2 - 4q + 5 > q^2 - 9 q + 21$ whenever $q \ge 4$.
Asymptotically,
both $\CountSR{q}{3}$ and $\CountCSR{q}{3}$
grow as $q^9$ as $q\To\infty$,
so the probability that a $3\times 3$ random matrix over $\fq$
is super-regular
(and therefore also contiguous super-regular)
tends to one as the field size grows.

For $4\times 4$ matrices,
the analysis appears to be significantly more complicated.
An open question in coding theory is to determine how many $[8,4]$ MDS
codes exist over $\fq$,
which is equivalent to determining the number
$\CountSR{q}{4}$
of $4\times 4$ super-regular matrices over $\fq$.
More broadly,
it is not known even if the number of $4\times 4$ (contiguous or not) super-regular
matrices is monotone increasing in the field size.

In Section~\ref{sec:computer-simulation},
we computationally determine the quantity $\CountSR{q}{4}$
for all fields of size at most $71$
(and also for sizes $83$ and $97$).
Then, using the numerically obtained results,
we show in Theorem~\ref{thm:not-quasi} that
the number of $4\times 4$ super-regular matrices over $\fq$
cannot be a quasi-polynomial of period less than $7$,
which in turn implies it cannot be a polynomial function of $q$.

For the contiguous case,
we computationally determine the quantity $\CountCSR{q}{4}$
for all fields of size at most $67$.
We discovered a polynomial of degree $16$
that explains all of our observed computational data.
Our Conjecture~\ref{conj:poly}
asserts that the polynomial is the true count $\CountCSR{q}{4}$.
Theorem~\ref{thm:4x4-poly} points out
that if $\CountCSR{q}{4}$ is indeed a polynomial,
then it must be the polynomial we found.

In Section~\ref{sec:Conjecture},
we conjecture that the $k=o(n)$ requirement in Theorem~\ref{thm:main-MDS}
can be dropped,
and then show that the convergences rates of 
$\ProbCSR{q}{k} \To e^{-\lambda}$ in
Theorem~\ref{thm:main-contiguous} 
and the rate of the conjectured convergence of
$\ProbSR{q}{k} \To e^{-\lambda}$ in
Conjecture~\ref{conj:converge}
are experimentally rapid.

Various technical lemmas are relegated to the Appendix in Section~\ref{sec:Appendix}.

\clearpage
%================================================
\section{Asymptotic probabilities of MDS codes and super-regular matrices}
\label{sec:MDS-codes}

This section analyzes the probabilities 
$\ProbMDS{q}{n}{k}$ and $\ProbSR{q}{k}$
as $q,n,k\To\infty$.

The following widely-used lemma is due to 
Arratia, Goldstein, and Gordon~\cite[Theorem 1]{Arratia-Goldstein-Gordon-1989}
(see also~\cite{Arratia-Goldstein-Gordon-1990})
and relies on the Chen-Stein Poisson approximation method.
The lemma shows that the probability that 
the sum of indicator random variables equals zero
is approximately the same as if the indicators were independent,
even when there are slight dependencies among them.
The lemma is applied to submatrices in the proof of Theorem~\ref{thm:main-MDS}(b),
where it helps control the effect of dependencies among overlapping submatrices.

\begin{lemma}
\label{lem:stein}
Let $I$ be a set and 
for each $\alpha \in I$ let $X_\alpha$ be a
Bernoulli random variable with 
$P(X_\alpha {=} 1) > 0$,
and let $D_\alpha \subseteq I$.
If
\begin{align*}
b_1 & = \sum_{\alpha \in I} \sum_{\beta \in D_\alpha} P(X_\alpha=1)P(X_\beta=1) \\
b_2 & = \sum_{\alpha \in I} \sum_{\substack{\beta \in D_\alpha \\ \beta \ne
\alpha}} P(X_\alpha {=} X_\beta {=} 1) \\
b_3 & = \sum_{\alpha \in I} 
      \E \left[ \left| 
        \E [X_\alpha - P(X_\alpha=1) \mid X_\beta : \beta \in I \setminus D_\alpha]
         \right| \right]\\
\mu &= \E\left[\sum_{\alpha \in I} X_\alpha\right],
\end{align*}
then
\begin{align}
	\left|P\left(\sum_{\alpha \in I} X_\alpha = 0\right)- e^{-\mu} \right| \le
	\min(1, \mu^{-1})(b_1 + b_2 + b_3). \label{eq:000}
\end{align}
\end{lemma}

Each set $D_\alpha$ in Lemma~\ref{lem:stein} is intended
to identify those elements $\beta$ of $I$ for which
$X_\alpha$ statistically depends on $X_\beta$.

We will exploit Lemma~\ref{lem:stein} in order to prove Theorem~\ref{thm:main-MDS}(c).
In Lemma~\ref{lem:stein}, 
we will set $I$ to be the set of
all $k \times k$ submatrices of a random $k\times n$ matrix $A$,
and for each $B \in I$ 
we will let $X_B$ be the indicator that $B$ is singular, i.e.,
\begin{align*}
X_B &= \begin{cases}
         1 & \text{if}\ B\ \text{is singular}\\
         0 & \text{if}\ B\ \text{is nonsingular}.
       \end{cases}
\end{align*}
Note that $|I| = \binom{n}{k}$ 
and the probability that any particular square submatrix of $A$ is singular
is asymptotically equal to $q^{-1}$ 
(by Lemma~\ref{lem:fullColRankProb}).
Then the event that $A$ is MDS is the event that 
all of its $k\times k$ submatrices are nonsingular, namely that
$\displaystyle\sum_{B \in I} X_B = 0$. 
The mean of this sum of indicators is
\begin{align*}
\E\left[\displaystyle\sum_{B \in I} X_B\right] 
&= |I| \cdot P(\text{A random}\ k\times k\ 
               \text{matrix over}\ \fq \ \text{is singular})\\
&\sim \binom{n}{k}\cdot q^{-1} 
\To \lambda 
\end{align*}
so Lemma~\ref{lem:stein} will give 
$P(A \text{ is MDS}) \approx e^{-\lambda}$ 
if we can show the upper bound in $\eqref{eq:000}$ tends to $0$.

In Lemma~\ref{lem:stein} we try to choose
$D_B$ so that \eqref{eq:000} is as small as possible. 
If $D_B$ excludes events in $I$ that strongly influence $X_B$, then
$b_3$ will end up being large. 
However, if $D_B$ is made too big, then 
$b_1, b_2$ will be big as well. 

For our case, we will choose
\begin{align*}
D_B = \{ B' \in I : B' \text{ contains the left-most column of }B\}.
\end{align*}
When one uses Lemma~\ref{lem:stein}, 
the typical choice of $D_B$ would contain only
events that are ``significantly'' dependent on $B$,
but there are two motivating reasons for our choice of $D_B$:
\begin{enumerate}
\item $|D_B| = \binom{n-1}{k-1} = \frac{k}{n} \binom{n}{k} = \frac{k}{n} |I|$. 
        Since we assume $|I|/q \To \lambda$ and $k/n \To 0$, we have
\begin{align*}
b_1 \approx b_2 \approx |D_B|\cdot |I| q^{-2} \To \lambda^2 \frac{k}{n} \To 0
\end{align*}
for $\lambda \ne \infty$. 
The fact that $b_1 \approx |D_B| \cdot |I|q^{-2}$
follows from Lemma~\ref{lem:fullColRankProb}. 
The fact that $b_2 \approx |D_B| \cdot |I|q^{-2}$ 
follows from Lemma~\ref{lem:pairwise}.

\item The left-most column of $B$ remains uniformly distributed on $\fq^k$
after conditioning on $\{X_{B'} : B' \in I \setminus D_B\}$.
This is crucial for controlling $b_3$. Lemma~\ref{lem:b3} bounds $b_3$.
\end{enumerate}

Removing the condition that $k/n \To 0$ in 
Theorem~\ref{thm:main-MDS}(c) requires choosing a smaller $D_B$ so that
$|D_B|/|I| \To 0$ even when 
$k/n$ does not tend to zero.
However, with such a choice of
$D_B$, no column (or entry) of $B$ would be independent of 
$\{X_{B'} : B' \in I \setminus D_B\}$. 
This would make it much harder to prove that $b_3 \To 0$.

Lemma~\ref{lem:fullColRankProb} 
below shows that the probability $m$ iid uniform vectors on
$\fq^d$ are linearly dependent is close to $q^{d-m-1}$. 
In particular, 
it shows that any square matrix whose entries are iid uniform over $\fq$ will be
singular with probability nearly $1/q$. 
Lemma~\ref{lem:fullColRankProb} will be used
throughout the proof of Theorem~\ref{thm:main-MDS}(c) and its supporting lemmas.

\begin{lemma} 
\label{lem:fullColRankProb}
Let $m \le d$ and let
$v_1, \dots, v_m$ be $d$-dimensional vectors whose elements
are iid uniform on the field $\fq$. 
Then for any $j< m$,
\begin{align*}
&\frac{q^{-(d-m)}}{q} \\
&\le
P(v_1, \dots, v_m \text{ are linearly dependent}\ 
 \big| \ v_1, \dots, v_j \text{ are linearly independent} )\\
&\le \frac{q^{-(d-m)}}{q-1}  \le 2\cdot \frac{q^{-(d-m)}}{q} .
\end{align*}
\end{lemma}

\begin{proof}
For any $j< m$ (including $j=0$), we have
\begin{align}
&P(v_1, \dots, v_m \text{ are linearly dependent}\ 
 | \ v_1, \dots, v_j \text{ are linearly independent} )\notag\\
&= 1 - \prod_{i=j}^{m-1} (1-q^{i-d}) \label{eq:022}\\
&\le 1 - \prod_{i=0}^{m-1} (1-q^{i-d}) \\
&\le 1 - \left(1 - \sum_{i=0}^{m-1} q^{i-d} \right) \label{eq:016}\\
&= q^{-(d-m)} \sum_{i=1}^{m} q^{-i} \notag\\
&\le q^{-(d-m)} \sum_{i=1}^\infty q^{-i} \notag\\
&= \frac{q^{-(d-m)}}{q-1} \notag\\
&\le \frac{q^{-(d-m)}}{q/2} \label{eq:020}
\end{align}
where
\eqref{eq:016} follows from Lemma~\ref{lem:inequality};
and
\eqref{eq:020} follows since $q \ge 2$.
Furthermore,
\begin{align*}
1 - \prod_{i=j}^{m-1} (1-q^{i-d})
& \ge 1 - (1-q^{m-1-d})
= \frac{q^{-(d-m)}}{q}.
\end{align*}
\end{proof}

Lemma~\ref{lem:fullColRankProb}
implies for the case $d=m$ and $j=0$ 
that the probability a random square matrix over $\fq$ is singular lies between
$\frac{1}{q}$ and $\frac{1}{q-1}$
regardless of the size of the matrix.

The next lemma is used in the proofs of
Lemma~\ref{lem:pairwise} and Lemma~\ref{lem:b3}.

\begin{lemma} 
Let $d, m, t$ be positive integers such that $1\le t < d$
and let 
$r_1, \dots, r_m \in \{1, \dots, d-t\}$.
Let $R_1, \dots, R_m, T \subseteq \fq^d$ 
each of whose components is chosen iid uniformly at random from $\fq$,
and such that $|T|=t$ 
and $|R_i| = r_i$ for each $i=1,\dots m$.
Then the events $\{R_i \cup T\ \text{is linearly dependent}\}$ for $1\le i \le m$
are conditionally independent given $T$ is linearly independent.
\label{lem:ConditionalIndependence}
\end{lemma}

\begin{proof}
Let $\mathcal{T}$ be the collection of all sets containing $t$
linearly independent vectors in $\fq^d$.
Then,
\begin{align}
&P(R_i\cup T \text{ is linearly independent } \forall i = 1, \dots, m\ 
    \Big|\ T \text{ is linearly independent})\notag\\
&= P\left(\bigcap_{i=1}^m 
  \{ R_i\cup T \text{ is linearly independent} \}
      \Big|\ T\in\mathcal{T} \right)\notag\\
&= \frac{1}{P(T\in\mathcal{T})}
   \sum_{\tau\in\mathcal{T}} 
   P\left( \bigcap_{i=1}^m 
         \{ R_i\cup T \text{ is linearly independent} \}
           \Big|\ T=\tau \right) P(T=\tau)\notag\\
&= \frac{1}{|\mathcal{T}|}
   \sum_{\tau\in\mathcal{T}} 
   P\left( \bigcap_{i=1}^m 
         \{ R_i\cup \tau \text{ is linearly independent} \} \right) \label{eq:98}\\
&= \frac{1}{|\mathcal{T}|}
   \sum_{\tau\in\mathcal{T}} 
   \prod_{i=1}^m 
   P\left( R_i\cup \tau \text{ is linearly independent} \right) \label{eq:99}\\
&= \frac{1}{P(T\in\mathcal{T})}
   \sum_{\tau\in\mathcal{T}} 
   \prod_{i=1}^m 
   P\left( R_i\cup T \text{ is linearly independent} \Big|\ T=\tau \right) 
    P(T=\tau)\notag\\
&=\prod_{i=1}^m
  P\left( R_i \cup T \text{ is linearly independent }  \ 
    \Big|\ T \text{ is linearly independent}\right) 
\end{align}
where 
\eqref{eq:98} follows since $P(T=\tau)$ is the same for all $\tau\in\mathcal{T}$;
and
\eqref{eq:99} follows since 
$\tau$ is a fixed linearly independent set
and the $R_i$'s are independent random variables.
The same argument works for any subcollection of $R_1, \dots, R_m$.
\end{proof}

The next lemma is used in the proof of Lemma~\ref{lem:b3}.

\begin{lemma} 
If $v_1, \dots, v_{d-1} \in \fq^d$ are linearly independent vectors
and $v_d$ is chosen uniformly at random from $\fq^d$,
then the probability that $v_1, \dots, v_d$ are
linearly dependent is $1/q$.
\label{lem:AllButOneVectorIndependent}
\end{lemma}

\begin{proof}
It follows immediately 
by setting $d=m=j+1$ in \eqref{eq:022} 
of the proof of Lemma~\ref{lem:fullColRankProb}.
\end{proof}

In what follows submatrices with some common columns will be analyzed.

We say that two submatrices $B$ and $B'$ of a matrix $A$ 
\textit{share a column}
if that column was derived from the same column in $A$ for both $B$ and $B'$
(see Figure~\ref{fig:two-submatrices}).

\begin{figure}[hht]
\begin{center}
\begin{tikzpicture}[scale=0.7]

\draw[-, ultra thick] (1,1)--(17,1)--(17,6)--(1,6)--cycle;

% Cols
\draw[fill=red] (14,1)--(15,1)--(15,6)--(14,6)--cycle;
\draw[fill=red] (5,1)--(6,1)--(6,6)--(5,6)--cycle;
\draw[fill=red] (1,1)--(2,1)--(2,6)--(1,6)--cycle;

\draw[fill=blue] (15,1)--(16,1)--(16,6)--(15,6)--cycle;
\draw[fill=blue] (4,1)--(5,1)--(5,6)--(4,6)--cycle;
\draw[fill=blue] (2,1)--(3,1)--(3,6)--(2,6)--cycle;

% ---------------------------------------------------------------------------
% Row-Col intersections
\draw[fill=red] (12,1)--(13,1)--(13,6)--cycle;
\draw[fill=red] (8,1)--(9,1)--(9,6)--cycle;
\draw[fill=blue] (12,6)--(13,6)--(12,1)--cycle;
\draw[fill=blue] (8,6)--(9,6)--(8,1)--cycle;

% labels

\draw[dashed, ultra thin] (1,2)--(17,2);
\draw[dashed, ultra thin] (1,3)--(17,3);
\draw[dashed, ultra thin] (1,4)--(17,4);
\draw[dashed, ultra thin] (1,5)--(17,5);
\draw[dashed, ultra thin] (1,6)--(17,6);

\draw[dashed, ultra thin] (2,1)--(2,6);
\draw[dashed, ultra thin] (3,1)--(3,6);
\draw[dashed, ultra thin] (4,1)--(4,6);
\draw[dashed, ultra thin] (5,1)--(5,6);
\draw[dashed, ultra thin] (6,1)--(6,6);
\draw[dashed, ultra thin] (7,1)--(7,6);
\draw[dashed, ultra thin] (8,1)--(8,6);
\draw[dashed, ultra thin] (9,1)--(9,6);
\draw[dashed, ultra thin] (10,1)--(10,6);
\draw[dashed, ultra thin] (11,1)--(11,6);
\draw[dashed, ultra thin] (12,1)--(12,6);
\draw[dashed, ultra thin] (13,1)--(13,6);
\draw[dashed, ultra thin] (14,1)--(14,6);
\draw[dashed, ultra thin] (15,1)--(15,6);

\node[black, above, scale=1] at (1.5,6) {$1$};
\node[black, above, scale=1] at (2.5,6) {$2$};
\node[black, above, scale=1] at (3.5,6) {$3$};
\node[black, above, scale=1] at (4.5,6) {$4$};
\node[black, above, scale=1] at (5.5,6) {$5$};
\node[black, above, scale=1] at (6.5,6) {$6$};
\node[black, above, scale=1] at (7.5,6) {$7$};
\node[black, above, scale=1] at (8.5,6) {$8$};
\node[black, above, scale=1] at (9.5,6) {$9$};
\node[black, above, scale=1] at (10.5,6) {$10$};
\node[black, above, scale=1] at (11.5,6) {$11$};
\node[black, above, scale=1] at (12.5,6) {$12$};
\node[black, above, scale=1] at (13.5,6) {$13$};
\node[black, above, scale=1] at (14.5,6) {$14$};
\node[black, above, scale=1] at (15.5,6) {$15$};
\node[black, above, scale=1] at (16.5,6) {$16$};

\node[black, left, scale=1] at (1,5.5) {$1$};
\node[black, left, scale=1] at (1,4.5) {$2$};
\node[black, left, scale=1] at (1,3.5) {$3$};
\node[black, left, scale=1] at (1,2.5) {$4$};
\node[black, left, scale=1] at (1,1.5) {$5$};

\end{tikzpicture}
\end{center}
\caption{
An example of a $5\times 16$ matrix $A$ with two different
$5\times 5$ submatrices (blue and red) that share two columns
(split blue/red).
}
\label{fig:two-submatrices}
\end{figure}
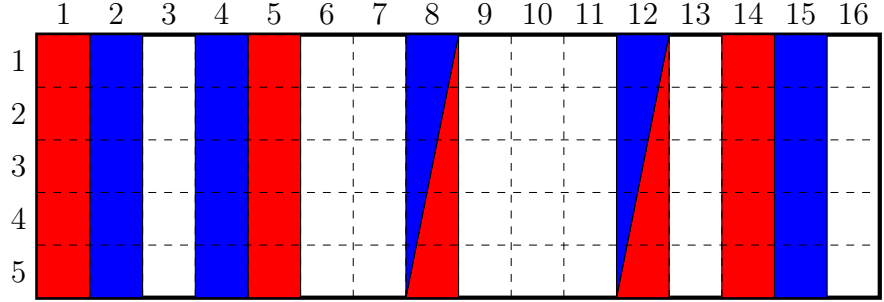

If $B$ and $B'$ are two non-overlapping $k\times k$ submatrices of a
random $k\times n$ matrix $A$ over $\fq$, then
the events that $B$ and $B'$ are singular are independent. 
This implies
$P(X_B=X_{B'}=1) = P(X_B=1)P(X_{B'}=1) \le 4q^{-2}$,
where the upper bound follows from Lemma~\ref{lem:fullColRankProb}.

The following lemma makes a similar statement even when the distinct 
submatrices do overlap.
Specifically, the lemma tells us that any distinct 
$k\times k$ submatrices $B$ and $B'$ of $A$ are ``not too'' dependent,
in the sense that $P(X_B=X_{B'}=1) \le 6q^{-2}$.
This fact is used to upper bound 
$b_2$ (from Lemma~\ref{lem:stein}) 
in the proof of Theorem~\ref{thm:main-MDS}(c).

The next lemma is used in the proofs of Theorem~\ref{thm:main-MDS}(b)(c).

\begin{lemma} \label{lem:pairwise}
Let $A$ be a random $k \times n$ matrix over the field $\fq$ and let
$B$ and $B'$ be distinct $k \times k$ submatrices of $A$
that share exactly $j$ columns.
Then,
\begin{align*}
%P(B\ \text{and}\ B'\ \text{are singular}) \le 2q^{j-k-1} + (q-1)^{-2}.
P(B\ \text{and}\ B'\ \text{are singular}) 
&\le \frac{2}{q^{k+1-j}} + \frac{1}{(q-1)^2}.
\end{align*}
\end{lemma}

\begin{proof}
Let $J$ be the event that the columns of $A$ shared by both $B$ and $B'$ are
linearly dependent. 
Then
\begin{align}
&P(X_B = X_{B'} = 1)\notag\\
&= P(X_B = X_{B'} = 1 | J) P(J) + P(X_B = X_{B'} = 1 | J^c) P(J^c)\notag\\
&= P(J) + P(X_B = X_{B'} = 1 | J^c) P(J^c)\notag\\
&= P(J) + P(X_B  = 1 | J^c)P(X_{B'} = 1 | J^c) P(J^c)\label{eq:0005}\\
&\le P(J) + P(X_B  = 1 | J^c)P(X_{B'} = 1 | J^c)\notag\\
&\le 2q^{j-k-1}+ P(X_B  = 1 | J^c)P(X_{B'} = 1 | J^c)\label{eq:0008}\\
&\le 2q^{j-k-1} + (q-1)^{-2}\label{eq:0007}
\end{align}
where
\eqref{eq:0005} follows since the events that 
   $B$ and $B'$ are singular are conditionally independent given $J^c$ 
   by Lemma~\ref{lem:ConditionalIndependence};
and
\eqref{eq:0008} and \eqref{eq:0007} follow from Lemma~\ref{lem:fullColRankProb}.
\end{proof}

The following lemma,
used in the proof of Theorem~\ref{thm:main-MDS}(c),
upper bounds the quantity $b_3$ in
Lemma~\ref{lem:stein}
in the context of nonsingular submatrices of a random matrix.

\begin{lemma} \label{lem:b3}
Let $A$ be a random $k \times n$ matrix over the field $\fq$ 
where $k\le n-2$
and let $I$ be the set of all $k \times k$ submatrices of $A$.
For any $B\in I$, 
let $X_B$ be the indicator that $B$ is singular,
and define
\begin{align*}
D_B &= \{ B' \in I : B' \text{ contains the left-most column of }B\}.
\end{align*}
Then
\begin{align*}
\sum_{B \in I} 
      \E \left[ \, \Big| 
        \, \E [X_B - P(X_B=1) \mid X_{B'} : B' \in I \setminus D_B] \,
         \Big| \, \right]
%	b_3 
&\le 5\binom{n}{k}q^{-2}.
\end{align*}
\end{lemma}

\begin{proof}
Fix $B \in I$.
Let $K$ denote the 
set of all events of the form 
\begin{align*}
\displaystyle\bigcap_{B' \in I \setminus D_B} \{X_{B'}=i_{B'}\}
\end{align*}
where $i_{B'}\in \{0,1\}$.
Note that 
$|I|=\binom{n}{k}$ and 
$|D_B| = \binom{n-1}{k-1}$
and exactly one of the 
$2^{|I\setminus D_B|}$
events in $K$ occurs.

By the definition of conditional expectation, 
\begin{align}
&\E[X_B - P(X_B=1) \mid X_{B'} : B' \in I \setminus D_B] \notag\\
& = \sum_{U \in K}\E[X_B - P(X_B=1)\mid U] \cdot \indicator{U} \notag\\
& = \sum_{U \in K}(P(X_B=1 \mid U) - P(X_B=1))\cdot\indicator{U}. \label{eq:00011}
\end{align}
Therefore since only one term in the summation in \eqref{eq:00011}
is nonzero, we get
\begin{align}
&\E\left[\Big|\, \E[X_B - P(X_B=1) \mid X_{B'} : 
                  B' \in I \setminus D_B]\Big|\, \right] 
    \notag\\
&= \sum_{U \in K}
\left|P(X_B=1 \mid U) - P(X_B=1)\right|\cdot P(U).
\label{eq:001}
\end{align}
We introduce two events that will be used to upper bound \eqref{eq:001}:
\begin{itemize}
\item Let $G$ be the event that the right-most $k-1$ columns of $B$ 
are linearly dependent.

\item Let $F$ be the event that $B'$ is singular 
for \emph{all} $B' \in I\setminus \{B\}$
containing the right-most $k-1$ columns of $B$.
\end{itemize}
Note that 
$G\subseteq F$ and
$F$ is a union of elements (i.e. of events) in $K$,
so every element of $K$ is a subset of either $F$ or $F^c$.
We now split the sum in \eqref{eq:001} based on whether 
$U$ is a subset of $F$ or $F^c$.
If $U\subseteq F^c$, then
\begin{align}
|P(X_B=1 \mid U) - P(X_B=1)| 
&= \left|\frac{1}{q} - P(X_B=1) \right| \label{eq:008}\\
&\le \frac{1}{q(q-1)}\label{eq:006}\\
&\le 2q^{-2} \label{eq:005}
\end{align}
where
\eqref{eq:008} follows since
$U\in K$ and the way $D_B$ is defined
together imply that the left-most column of $B$ is independent of $U$,
and since $F^c \subseteq G^c$, we have $U\subseteq G^c$, so
$P(X_B=1 \mid U) = 1/q$, 
by Lemma~\ref{lem:AllButOneVectorIndependent};
\eqref{eq:006} follows since $P(X_B=1)\le (q-1)^{-1}$ 
              from Lemma~\ref{lem:fullColRankProb};
and
\eqref{eq:005} follows since $q\ge 2$.

Thus,
\begin{align}
& \sum_{U \in K}\left|P(X_B=1 \mid U) - P(X_B=1)\right|P(U) \notag\\
& = \sum_{\substack{U \in K\\ U \subseteq F}} 
    \left|P(X_B=1 \mid U) - P(X_B=1)\right|P(U) \notag\\
&\ \ \   + \sum_{\substack{U \in K\\ U \subseteq F^c}}
    \left|P(X_B=1 \mid U) - P(X_B=1)\right|P(U) \notag\\
& \le \sum_{\substack{U \in K\\ U \subseteq F}} 
   P(U) + 2q^{-2} \sum_{\substack{U \in K\\ U \subseteq F^c}}  P(U) 
   \label{eq:017}\\
& \le P(F) + 2q^{-2} \label{eq:004}
\end{align}
where 
\eqref{eq:017} follows from
   $|P(X_B=1 \mid U) - P(X_B=1)| \le 1$ for $U\subseteq F$,
   and 
   from \eqref{eq:005} for $U\subseteq F^c$;
and \eqref{eq:004} follows since the events in $K$ are disjoint.

If $H_1, \dots, H_{n-k}$ denote the submatrices of $A$ 
containing the right-most $k-1$ columns of $B$, then
\begin{align}
P(F \mid G^\complement) 
&= P(H_1, \dots, H_{n-k}\ \text{are singular} \mid G^c) \notag\\
&= (P(H_1 \ \text{is singular} \mid G^c))^{n-k}  \label{eq:021}\\
&= q^{-(n-k)} \label{eq:023}\\
&\le q^{-2}. \label{eq:003}
\end{align}
where 
\eqref{eq:021} follows since Lemma~\ref{lem:ConditionalIndependence}
    implies the events that $H_i$ is singular, 
    for $i=1, \dots, n-k$,
    are conditionally independent given $G^c$;
\eqref{eq:023} follows since $P(H_i\ \text{is singular} \mid G^c) = 1/q$
    for all $i$;
and
\eqref{eq:003} follows from $k \le n-2$.
We have
\begin{align}
P(F) &= P(F \mid G)P(G) + P(F \mid G^\complement)P(G^\complement) \notag\\
&= P(G) + P(F \mid G^\complement)P(G^c) \label{eq:015}\\
&\le 2q^{-2} + P(F \mid G^\complement) \label{eq:019}\\
&\le 2q^{-2} + q^{-2} = 3q^{-2} \label{eq:007}
\end{align}
where 
\eqref{eq:015} follows from $G\subseteq F$;
\eqref{eq:019} follows from Lemma~\ref{lem:fullColRankProb}
               and $P(G^c)\le 1$;
and
\eqref{eq:007} follows from \eqref{eq:003}.

The proof up to this point has assumed a fixed submatrix $B$.
Now we sum over such $B$ to obtain the claimed upper bound.
Combining 
\eqref{eq:001}, 
\eqref{eq:004}, and 
\eqref{eq:007} gives
\begin{align*}
\E \left[\,\Big|\,\E\left[X_B - P(X_B=1) \mid X_{B'} : B' \in I
\setminus D_B\right]\,\Big|\,\right] \le P(F) + 2q^{-2} \le 5q^{-2}
\end{align*}
which implies
\begin{align*}
\sum_{B \in I} 
      \E \left[ \left| 
        \E [X_B - P(X_B=1) \mid X_{B'} : B' \in I \setminus D_B]
         \right| \right]
&\le \sum_{B \in I} 5q^{-2}
= 5q^{-2}\cdot |I|
= 5\binom{n}{k}q^{-2}.
\end{align*}
\end{proof}

% ---------------------------------------------------------------------------
\begin{theorem} \label{thm:main-MDS}
Let $C$ be a linear $[n,k]$ code chosen uniformly at random
over a field of size $q$,
where $n,q,k \To\infty$.
\begin{itemize}
\item[(a)] If $\frac{1}{q}\binom{n}{k}\To 0$, then $P(C \text{ is MDS}) \To 1$.
\item[(b)] If $\frac{1}{q}\binom{n}{k}\To \infty$, then $P(C \text{ is MDS}) \To 0$.
\item[(c)] If $\frac{1}{q}\binom{n}{k}\To \lambda\in (0,\infty)$ and
      $k/n \To 0$, then $P(C \text{ is MDS}) \To e^{-\lambda}$.
\end{itemize}
\end{theorem}

\begin{proof}
\ \\

Let $G$ be a $k\times n$ random matrix over $\fq$
and
let $I$ be the set of $k \times k$ submatrices of $G$.
For each $B \in I$ let $X_B$ be the indicator random variable of $B$ being singular.
Lemma~\ref{lem:multiple-defs-of-code} implies that 
the probabilities of the code $C$ being MDS 
and the matrix $G$ being an MDS generator matrix are asymptotically the same.

\noindent\textbullet\ Part (a):\\
The proof follows from a result in 
\cite{Ghorpade-Lachaud-2001}.
Here is an elementary proof.
\begin{align*}
P(G\ \text{is MDS})
&= 1 - P\left(\bigcup_{B\in I} \{X_B = 1\} \right)
\ge 1 - \sum_{B\in I} P(X_B = 1)
\ge 1 - \frac{1}{q-1}\binom{n}{k}
\To 1
\end{align*}
where the first equality follows from
Lemma~\ref{lem:MDS-conditions}(b),
and the last inequality follows from
Lemma~\ref{lem:fullColRankProb}
and $|I| = \binom{n}{k}$.

% ---------------------------------------------------------------------------

\noindent\textbullet\ Part (b):\\
Let $\mu = \E \left[ \displaystyle\sum_{B\in I} X_B \right]$.
If $j(B, B')$ is the number of columns shared by $B$ and $B'$,
then
\begin{align}
\Var\left(\sum_{B \in I} X_B\right)
 &= \sum_{B\in I} \Var(X_B) + \sum_{B\in I} \sum_{B'\in I\setminus\{B\}} 
     \Cov(X_B, X_{B'}) \notag\\
 &= \sum_{B\in I} P(X_B=1)P(X_B=0) + \sum_{B\in I} \sum_{B'\in I\setminus\{B\}} 
     \Cov(X_B, X_{B'}) \notag\\
 &\le \sum_{B\in I} P(X_B=1) + \sum_{B\in I} \sum_{B'\in I\setminus\{B\}} 
     \Cov(X_B, X_{B'}) \notag\\
 & = \mu + \sum_{B\in I} \sum_{B'\in I\setminus\{B\}} \Cov(X_B, X_{B'}) \notag\\
 & = \mu + \sum_{B\in I} \sum_{B'\in I\setminus\{B\}} 
      \left( P(X_B {=} X_{B'} {=}1) - P(X_B{=}1)P(X_{B'}{=}1) \right) \notag\\
 & \le \mu + \sum_{B\in I} \sum_{B'\in I\setminus\{B\}} 
       \left( 2q^{j(B, B')-k-1} + (q-1)^{-2}  - q^{-2}\right)  \label{eq:9492}\\
 & = \mu + \sum_{B\in I} \sum_{B'\in I\setminus\{B\}} 
      \left( 2q^{j(B, B')-k-1} + \frac{2q-1}{(q-1)^2q^2} \right) \notag\\
 & \le \mu + \sum_{B\in I} \sum_{B'\in I\setminus\{B\}} 
     \left( 2q^{j(B, B')-k-1} + 8q^{-3} \right) \label{eq:9493}
\end{align}
where 
\eqref{eq:9492} follows from 
         Lemma~\ref{lem:pairwise} and Lemma~\ref{lem:fullColRankProb};
and
\eqref{eq:9493} follows since 
$8q^{-3} - \frac{2q-1}{(q-1)^2q^2} = 
 \frac{6q^2 -15q+8}{(q-1)^2q^3}\ge 0$ for all $q\ge 2$.

Let $N_j$ be the number of ordered pairs of 
distinct submatrices in $I$ sharing $j$ columns. 
Then, from \eqref{eq:9493},
\begin{align}
\Var\left(\sum_{B \in I} X_B\right)
&\le \mu + \sum_{j = 0}^{k-1}N_j(2q^{j-k-1} + 8q^{-3}) \notag\\
&\le \mu + N_{k-1}(2q^{-2} + 8q^{-3}) + 10q^{-3} \sum_{j=0}^{k-2} N_j 
      \label{eq:9494}\\
&\le \mu + N_{k-1}(2q^{-2} + 8q^{-3}) + 10q^{-3} \binom{n}{k}^2 
     \label{eq:9495}  \\
&= \mu + nk\binom{n}{k}(2q^{-2} + 8q^{-3}) + 10q^{-3}\binom{n}{k}^2 
     \label{eq:9496} \\
&\le \mu + nk\binom{n}{k}10q^{-2} + 10q^{-3} \binom{n}{k}^2  \notag\\
&\le \mu + 10 \mu^2 \frac{nk}{\binom{n}{k}} + 10 \frac{\mu^2}{q} 
          \label{eq:9497} 
\end{align}
where
\eqref{eq:9494} follows since $q^{j-k-1} \le q^{-3}$ for all $j\le k-2$;
\eqref{eq:9495} follows from 
   $\displaystyle\sum_{j=0}^k N_j = |I|^2 = \binom{n}{k}^2$;
\eqref{eq:9496} follows from 
    $N_{k-1} = \binom{n}{k} \binom{k}{k-1} \binom{n-k}{1} \le nk \binom{n}{k}$;
and
\eqref{eq:9497} follows from $\mu \ge \frac{1}{q}\binom{n}{k}$.

\begin{align}
P(G\ \text{is MDS})
&= P\left(\sum_{B \in I} X_B = 0\right)\notag\\
&\le  P\left(\left|\mu - \sum_{B \in I} X_B \right| \ge \mu\right).
      \label{eq:029}\\
& \le \frac{1}{\mu^2} \cdot\Var\left(\displaystyle\sum_{B \in I} X_B\right)  
      \label{eq:9498}\\
& \le 10 \frac{nk}{\binom{n}{k}} + 10\frac{1}{q} +  \frac{1}{\mu} 
      \label{eq:2987}
\end{align}
where
\eqref{eq:029} follows since
   $\displaystyle\left|\mu - \sum_{B \in I} X_B \right| \ge \mu$
   whenever
   $\displaystyle\sum_{B \in I} X_B = 0$;
\eqref{eq:9498} follows from Chebyshev's inequality;
and
\eqref{eq:2987} follows from \eqref{eq:9497}.

Next, we use the assumption that $\lambda = \infty$.
Since $\mu \ge \frac{1}{q}\binom{n}{k}\To\lambda=\infty$
and $q\To\infty$ by assumption,
the terms $10/q$ and $1/\mu$ in \eqref{eq:2987} each tend to zero,
so to show $P(G\ \text{is MDS}) \To 0$
it suffices to show $nk/\binom{n}{k} \To 0$.

Also, we must have $\binom{n}{k}\To \infty$ and therefore $n\To\infty$.
If $3 \le k\le n-3$, then $\binom{n}{k}$ is at least cubic in $n$, in which case
$nk/\binom{n}{k} \To 0$.
If $k=1$, then 
$P(G\ \text{is MDS})$
is the probability all $n$ elements of $G$ are nonzero, namely 
$(1-q^{-1})^n \sim e^{-\frac{n}{q}} = e^{-\frac{1}{q}\binom{n}{k}} \To 0$.
If $k=2$, then
$nk/\binom{n}{k} = 2n/\binom{n}{2} = \frac{4}{n-1}\To 0$.

If $k=n-1$,
then by Lemma~\ref{lem:multiple-defs-of-code}
the probability that $G$ is MDS is asymptotically the same 
as the probability that an $(n-1)\times 1$ random matrix over $\fq$ is super-regular,
which equals the probability that a $1 \times (n-1)$ 
random matrix over $\fq$ is super-regular.
Using Lemma~\ref{lem:multiple-defs-of-code} again,
this probability is asymptotically 
equal to the probability that a $1\times n$ matrix is MDS.
Therefore this case reduces to the $k=1$ case in the previous paragraph.
The same argument also works for $k=n-2$.
% ---------------------------------------------------------------------------

\noindent\textbullet\ Part (c):\\
For each $B\in I$,
define
\begin{align*}
D_B = \{ B' \in I : B' \text{ contains the left-most column of }B\}.
\end{align*}
Let $b_1, b_2, b_3, \mu$ be the corresponding quantities 
defined in Lemma~\ref{lem:stein}. 
Then 
\begin{align}
\left|P(G \text{ is MDS}) - e^{-\lambda}\right| 
&=\left|P\left(\sum_{B \in I} X_B = 0\right) - e^{-\lambda}\right| \notag\\
&\le \left|P\left(\sum_{B \in I} X_B = 0\right) - e^{-\mu}\right| 
           + |e^{-\mu} - e^{-\lambda}|\notag\\
&\le \min(1, \mu^{-1})(b_1 + b_2 + b_3) + |e^{-\mu} - e^{-\lambda}|
\label{eq:012}
\end{align}
where 
\eqref{eq:012} follows from Lemma~\ref{lem:stein}.
Using Lemma~\ref{lem:fullColRankProb} and
$|I|=\binom{n}{k}$ we have
\begin{align*}
\mu 
&= \displaystyle E \left[\sum_{B\in I} X_B \right]
 = |I|\cdot P(X_B = 1)\\
\therefore \frac{1}{q}\binom{n}{k} 
&\le 
\mu
\le \frac{1}{q-1}\binom{n}{k}.
\end{align*}
By assumption, 
$\frac{1}{q}\binom{n}{k} \To \lambda$,
so $\mu\To\lambda$ 
and therefore $|e^{-\mu} - e^{-\lambda}|\To 0$.

We next show that $\min(1, \mu^{-1})(b_1 + b_2 + b_3) \To 0$.
\begin{align}
b_1 
&= \sum_{B'\in I} \sum_{B\in D_B} P(X_B=1)P(X_{B'}=1)\notag\\
&\le |I|\cdot \max_{B \in I} |D_B| \cdot 4q^{-2} \label{eq:025}\\
&= 4 \binom{n}{k}^2 q^{-2}\cdot \frac{k}{n} \label{eq:010}
\end{align}
where
\eqref{eq:025} follows since Lemma~\ref{lem:fullColRankProb} implies
              $P(X_B=1)P(X_{B'}=1) \le 2q^{-1} \cdot 2q^{-1} = 4q^{-2}$;
and
\eqref{eq:010} follows from $|I| = \binom{n}{k}$ and 
              $|D_B| = \binom{n-1}{k-1} = \binom{n}{k} \frac{k}{n}$.
\begin{align}
b_2 
& = \sum_{B' \in I} \sum_{\substack{B \in D_{B'} \\ B \ne B'}} 
    P(X_B = X_{B'} = 1) \notag\\
&\le |I| \cdot \max_{B \in I}|D_B \setminus \{B\}| \cdot 6q^{-2} \label{eq:026}\\
&\le 6 \binom{n}{k}^2q^{-2} \cdot\frac{k}{n} \label{eq:011}
\end{align}
where
\eqref{eq:026} follows since 
     $B\ne B'$ implies that $B$ and $B'$ share at most $k-1$ columns,
%     $D_B\setminus \{B\} = \emptyset$ when $k=1$, and otherwise
     $q\ge 2$ implies $(q-1)^{-2} \le 4q^{-2}$,
     and Lemma~\ref{lem:pairwise} therefore implies 
\begin{align*}
P(X_B = X_{B'} = 1) 
&\le \frac{2}{q^{k+1-(k-1)}} + \frac{1}{(q-1)^2} 
 \le 6q^{-2} ;
\end{align*}
and
\eqref{eq:011} follows since $|D_B\setminus \{B\}|=\binom{n}{k}\frac{k}{n}-1$
               for all $B\in I$.

Since $k/n\To 0$,
for large enough $n$ we have $k\le n-2$ 
and therefore
Lemma~\ref{lem:b3} tells us that
\begin{align}
b_3 \le 5\binom{n}{k} q^{-2}.
\label{eq:024}
\end{align}
Furthermore, Lemma~\ref{lem:fullColRankProb} gives
\begin{align}
\mu^{-1} &= \left(\sum_{B \in I} P(X_B=1)\right)^{-1} \le q/|I|
 = q\binom{n}{k}^{-1}. 
\label{eq:013}
\end{align}

Then
\begin{align}
|P(G \text{ is MDS}) - e^{-\mu}| 
&\le \min(1, \mu^{-1})\left(4\binom{n}{k}^2 q^{-2} \cdot
\frac{k}{n} + 6 \binom{n}{k}^2 q^{-2} \cdot\frac{k}{n} + 5 \binom{n}{k}q^{-2}\right) 
  \label{eq:027}\\
&= \min(1, \mu^{-1})
    \left(10\binom{n}{k}^2 q^{-2} \cdot\frac{k}{n} + 5\binom{n}{k}q^{-2}\right)\notag \\
&\le 10\binom{n}{k} q^{-1} \cdot\frac{k}{n} + 5q^{-1}
\label{eq:014}
\end{align}
where
\eqref{eq:027} follows from
              \eqref{eq:012}, \eqref{eq:010}, \eqref{eq:011}, and \eqref{eq:024};
and
\eqref{eq:014} follows using \eqref{eq:013}.

If $\lambda<\infty$, 
then $\binom{n}{k}/q$ is bounded,
so the assumptions $k/n \To 0$ and $q \To \infty$ imply from \eqref{eq:014}
that $|P(G \text{ is MDS}) - e^{-\mu}| \To 0$,
and therefore,
by \eqref{eq:012}
we have $P(G \text{ is MDS}) \To e^{-\lambda}$.

If $\lambda = \infty$,
then the result is a special case of
Theorem~\ref{thm:main-MDS}(b).
\end{proof}

Theorem~\ref{thm:main-MDS} requires $n,q,k\To\infty$,
but some generalizations hold.
Parts (a),(b),(c) all hold if $k$ is bounded.
For part (a), $n$ need not tend to infinity
(the existing proof covers this).
For part (b), $q$ can be bounded.
To see this, suppose $q$ is bounded and let $\hat{q}$ be the maximum value of $q$.
Using Lemma~\ref{lem:AllButOneVectorIndependent},
for each $i\ge k$,
if the first $k-1$ columns of a $k\times n$ generator matrix are linearly independent, 
then the probability that the $i$th column is not in the span the first $k-1$
columns is at most $1-\hat{q}^{-1}$,
so the probability of getting $k$ linearly dependent columns tends to $1$.

Also, note that the premise of part (c) can be strengthened to allow
$\lambda\in [0,\infty]$,
but then the conclusion for the cases 
$\lambda=0$ and $\lambda=\infty$ would be weaker than 
parts (a) and (b) due to the extra $k/n\To 0$ condition.
Furthermore, the condition $k/n\To 0$ in part (c) can be changed to $k/n\To 1$
and the result still holds, by a similar argument.

The results in Theorem~\ref{thm:main-MDS} 
also hold for the probability that 
a random $k\times (n-k)$ matrix $A$ is super-regular
by Lemma~\ref{lem:MDS-conditions}.
The following result is a special case for square matrices.

\begin{corollary} 
Let $A$ be a random $k \times k$ matrix over a field of size $q$,
where $q, k \To \infty$.
\begin{enumerate}
\item [(a)] If\ \ $\frac{4^k/\sqrt{k}}{q} \To 0$, 
 then $P(A \text{\ is super-regular}) \To 1$.

\item [(b)] If\ \ $\frac{4^k/\sqrt{k}}{q} \To \infty$, 
 then $P(A \text{\ is super-regular}) \To 0$.
\end{enumerate}
\label{cor:main-superregular}
\end{corollary}

\begin{proof}
Let $I$ be the $k\times k$ identity matrix.
Since $\binom{2k}{k}\sim 4^k/\sqrt{\pi k}$ as $k\To\infty$,
setting $n=2k$ in Theorem~\ref{thm:main-MDS} implies that the probability 
the code generated by $k\times 2k$ matrix $[I|A]$ is MDS tends to $1$.
Also, Lemma~\ref{lem:MDS-conditions} % (a)(c)
implies that $A$ is super-regular if and only if $[I|A]$ is an MDS generator matrix.
\end{proof}

Thus, the rate of growth $q \sim 4^k/\sqrt{\pi k}$
is an asymptotic threshold on
the frequency of super-regular random matrices.

% ---------------------------------------------------------------------------

\clearpage

\section{Asymptotic probability of contiguous super-regular matrices}
\label{sec:Contiguous-Super-Regular-Square-Matrices}

This section analyzes the probability $\ProbCSR{q}{k}$
as $q,k\To\infty$.

The following definitions and lemmas will be used to prove
Lemma~\ref{lem:strong-inf} and Lemma~\ref{lem:Cauchy2},
which together prove the main result of this section,
Theorem~\ref{thm:main-contiguous}.

A submatrix of matrix $A$ is said to be \textit{anchored at $(i,j)$} 
if its bottom-most row and
right-most column are formed from the $i$th row and $j$th column of $A$,
respectively.
For each $k \le \minij$,
let $A_{i,j}(k)$ denote the $k\times k$ contiguous submatrix of $A$ anchored at $(i,j)$.

Suppose $B$ is a $k\times m$ matrix with $k,m \ge 2$.
Let $\bar{B}$ denote the $(k-1)\times (m-1)$ submatrix of $B$
obtained by removing the right-most column
and bottom-most row of $B$.
If $u^\transpose$ is the bottom-most row of $B$ without the right-most entry,
and $v$ is the right-most column of $B$ without the bottom-most entry,
then the triple $(\bar{B}, u, v)$ is called the \textit{corner decomposition}%
\footnote{
This is also referred to as a \textit{bordered matrix}
in~\cite{HornJohnson-book}. %pp. 24-26 
}
of $B$.
In algebraic expressions,
$\bar{B}$ will be treated as a matrix,
and $u$ and $v$ as column vectors.

\begin{figure}[hht]
\begin{center}
\begin{tikzpicture}[scale=0.7]

\draw[-, ultra thick] (1,1)--(10,1)--(10,10)--(1,10)--cycle;

% Rows
\draw[fill=blue] (1,5)--(1,9)--(10,9)--(10,5)--cycle;

% Cols
\draw[fill=red] (4,1)--(8,1)--(8,10)--(4,10)--cycle;

% Row-Col intersections
\draw[fill=red!70!blue] (4,5)--(4,9)--(8,9)--(8,5)--cycle;

% labels
\node[black,  right, scale=1.5] at (10,2.5) {$(5,7)$};
\draw[->, black, ultra thick] (11,3)--(7.5,5.5);

\draw[dashed, ultra thin] (1,1)--(10,1);
\draw[dashed, ultra thin] (1,2)--(10,2);
\draw[dashed, ultra thin] (1,3)--(10,3);
\draw[dashed, ultra thin] (1,4)--(10,4);
\draw[dashed, ultra thin] (1,5)--(10,5);
\draw[dashed, ultra thin] (1,6)--(10,6);
\draw[dashed, ultra thin] (1,7)--(10,7);
\draw[dashed, ultra thin] (1,8)--(10,8);
\draw[dashed, ultra thin] (1,9)--(10,9);

\draw[dashed, ultra thin] (1,1)--(1,10);
\draw[dashed, ultra thin] (2,1)--(2,10);
\draw[dashed, ultra thin] (3,1)--(3,10);
\draw[dashed, ultra thin] (4,1)--(4,10);
\draw[dashed, ultra thin] (5,1)--(5,10);
\draw[dashed, ultra thin] (6,1)--(6,10);
\draw[dashed, ultra thin] (7,1)--(7,10);
\draw[dashed, ultra thin] (8,1)--(8,10);
\draw[dashed, ultra thin] (9,1)--(9,10);

\node[black, above, scale=1] at (1.5,10) {$1$};
\node[black, above, scale=1] at (2.5,10) {$2$};
\node[black, above, scale=1] at (3.5,10) {$3$};
\node[black, above, scale=1] at (4.5,10) {$4$};
\node[black, above, scale=1] at (5.5,10) {$5$};
\node[black, above, scale=1] at (6.5,10) {$6$};
\node[black, above, scale=1] at (7.5,10) {$7$};
\node[black, above, scale=1] at (8.5,10) {$8$};
\node[black, above, scale=1] at (9.5,10) {$9$};

\node[black, left, scale=1] at (1,9.5) {$1$};
\node[black, left, scale=1] at (1,8.5) {$2$};
\node[black, left, scale=1] at (1,7.5) {$3$};
\node[black, left, scale=1] at (1,6.5) {$4$};
\node[black, left, scale=1] at (1,5.5) {$5$};
\node[black, left, scale=1] at (1,4.5) {$6$};
\node[black, left, scale=1] at (1,3.5) {$7$};
\node[black, left, scale=1] at (1,2.5) {$8$};
\node[black, left, scale=1] at (1,1.5) {$9$};

\node[black, above, scale=1] at (15.5,8) {$4$};
\node[black, above, scale=1] at (16.5,8) {$5$};
\node[black, above, scale=1] at (17.5,8) {$6$};
\node[black, above, scale=1] at (18.5,8) {$7$};

\node[black, left, scale=1] at (15,7.5) {$2$};
\node[black, left, scale=1] at (15,6.5) {$3$};
\node[black, left, scale=1] at (15,5.5) {$4$};
\node[black, left, scale=1] at (15,4.5) {$5$};

\node[black,  scale=1.5] at (16.5,6.5) {$\bar{B}$};
\node[black,  scale=1.5] at (16.5,4.5) {$u^\transpose$};
\node[black,  scale=1.5] at (18.5,6.5) {$v$};
\node[black,  scale=1.5] at (18.5,1.5) {$(5,7)$};
\draw[->, black, ultra thick] (18.5,2)--(18.5,4.5);
\draw[-, ultra thick, green] (19,4)--(19,8)--(15,8)--(15,4)--cycle;
\draw[thin] (18,4)--(18,8);
\draw[thin] (15,5)--(19,5);

\end{tikzpicture}
\end{center}
\caption{
An example of a $9\times 9$ square matrix (black) $A$ with
a $4\times 4$ contiguous submatrix $B$ anchored at $(5,7)$
formed by deleting from $A$ all but its red columns and blue rows.
The corner decomposition
$(\bar{B}, u, v)$ of $B$
is shown to the right in green.
}
\label{Fig:squares}
\end{figure}
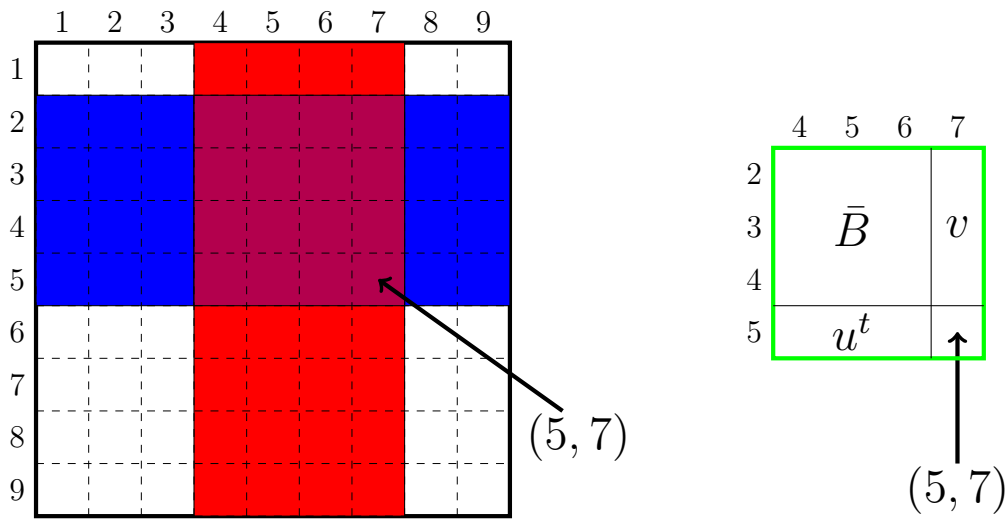

The following lemma is used in the proof of Lemma~\ref{lem:bad_value_condition}.

\begin{lemma}(e.g., \cite[pp. 24-26]{HornJohnson-book})
\label{lem:blockDeterminant}
If $U$ is a nonsingular $m \times m$ matrix,
$V$ is an $m \times n$ matrix,
$W$ is an $n \times m$ matrix,
and $Z$ is an $n \times n$ matrix, then
\begin{align*}
\det 
\begin{bmatrix} U & V \\ W & Z \end{bmatrix}
= \det(U) \det(Z - WU^{-1}V).
\end{align*}
\end{lemma}

\begin{definition}
If $B$ is a square matrix over $\fq$
that becomes singular when its bottom-right element is replaced by $y\in \fq$,
then $y$ is called a \textit{bad value} of $B$.
\end{definition}

The following lemma is used in the proofs of 
Lemma~\ref{lem:two-bad-values-strong2} and
Lemma~\ref{lem:strong-inf}.

\begin{lemma}
If $B$ is a square matrix
with corner decomposition $(\bar{B}, u, v)$
and $\bar{B}$ is invertible,
then the quantity
$X(B) = u^\transpose\bar{B}^{-1} v$ is the unique bad value of $B$.
\label{lem:bad_value_condition}
\end{lemma}

\begin{proof}
If $B$ is $k\times k$, then applying
Lemma~\ref{lem:blockDeterminant}
with $m = k - 1$, $n = 1$, $U = \bar{B}$,
$V = v$, $W = u^\transpose$, and $Z = B_{k,k}$ gives
\begin{align*}
\det B = \left( B_{k,k} - X(B) \right)\cdot \det \bar{B}
\end{align*}
so $B$ is singular if and only if
its bottom-right element equals its bad value, i.e.,
$B_{k,k} = X(B)$.
\end{proof}

Note that in the trivial case of a $1\times 1$ matrix,
the matrix is also singular if and only if
its bottom-right element equals its bad value (i.e. $0$).
The uniqueness of the bottom-right element guaranteed 
in Lemma~\ref{lem:bad_value_condition}
disappears if the submatrix invertibility requirement is relaxed.
Specifically, it can be seen using a cofactor determinant calculation 
that if the top-left $(k-1)\times (k-1)$ submatrix of 
a $k\times k$ matrix over $\fq$ is singular,
then the determinant of the $k\times k$ matrix does not depend on its bottom-right entry,
and therefore
there are either no bad values or else every field element is a bad value.

% ---------------------------------------------------------------------------

%-------------------------------

For any location $(i,j)$ in a random matrix over $\fq$,
let $\theta_{i,j}$ be the event that
every square contiguous submatrix anchored at $(i,j)$ is nonsingular,
and let $\Gamma_{i,j}$ be the event that
$\theta_{u,v}$ occurs for all $(u,v)$ preceeding $(i,j)$ in raster-scan order.
Also let $\Gamma_{1,1}$ be the sure event.

If the random values of a square matrix
are selected in raster-scan order
and at some point all contiguous submatrices already created are invertible,
then it turns out that the next random value added to the matrix is unlikely to
be the bad value of more than one new submatrix.
The following lemma makes this precise by
giving an upper bound on the probability that two matrices
anchored at the same location have the same bad value,
given that all previous contiguous square submatrices 
in raster-scan order are nonsingular.
The lemma is used in the proof of Lemma~\ref{lem:Cauchy2}.

%---------------------------------------------------------------------------
\begin{lemma}
If $k^3/q \To \lambda\in [0,\infty)$ as $k,q\To\infty$,
then there exist sufficiently large $k$ and $q$ such that
for any random $k \times k$ matrix $A$ over $\fq$,
and any location $(i,j)$ in $A$,
the conditional probability
given $\Gamma_{i,j}$
that the bad values of all 
square contiguous submatrices of $A$ anchored at $(i,j)$
are distinct
is at least $1 - \frac{2(\minij)^2}{q}$.
\label{lem:two-bad-values-strong2}
\end{lemma}

\begin{proof}
First note that if $\minij=1$,
then there is only one square contiguous submatrix anchored at $(i,j)$,
so the lemma is trivially true.
Thus, we will now assume $\minij\ge 2$.

Let $k$ and $q$ be large enough so that 
\begin{align}
\left( 1 - \frac{k}{q} \right)^k \ge \frac{1}{2}. \label{eq:804}
\end{align}
This can be done since 
$\frac{k^3}{q}\To\lambda<\infty$
as
$k,q\To\infty$,
and therefore
$\left( 1 - \frac{k}{q} \right)^k \To 1$
by Lemma~\ref{lem:convergence}.

\newcommand{\newj}{j'}
\newcommand{\bottomcorner}{{i,\newj}}
For each $m=1, \dots, \minij$,
let $B_m$ be the unique $m\times m$ submatrix anchored at $(i,j)$
(see Figure~\ref{Fig:Submatrices}).
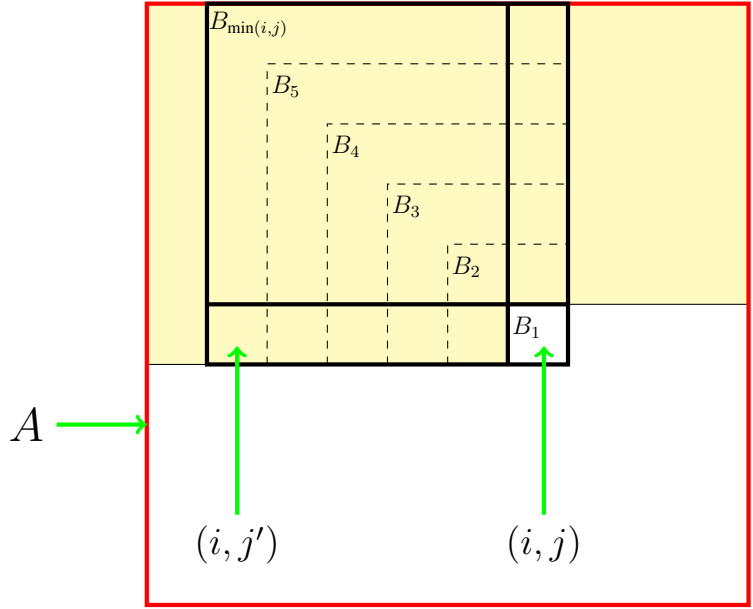
\begin{figure}[ht]
\begin{center}

\begin{tikzpicture}[scale=0.8]
\draw[fill=yellow!30] (-1,6)--(9,6)--(9,1)--(5,1)--(5,0)--(-1,0)--cycle;
\draw[red, ultra thick] (-1,6)--(9,6)--(9,-4)--(-1,-4)--cycle;
\draw[-, ultra thick] (0,0)--(6,0)--(6,6)--(0,6)--cycle;
\draw[-, ultra thick] (0,1)--(6,1);
\draw[-, ultra thick] (5,0)--(5,6);

\draw[dashed, thin] (4,0)--(4,2)--(6,2);
\draw[dashed, thin] (3,0)--(3,3)--(6,3);
\draw[dashed, thin] (2,0)--(2,4)--(6,4);
\draw[dashed, thin] (1,0)--(1,5)--(6,5);

\node[black, above left, scale=0.8] at (5.7,0.3) {$B_1$};
\node[black, above left, scale=0.8] at (4.7,1.3) {$B_2$};
\node[black, above left, scale=0.8] at (3.7,2.3) {$B_3$};
\node[black, above left, scale=0.8] at (2.7,3.3) {$B_4$};
\node[black, above left, scale=0.8] at (1.7,4.3) {$B_5$};
\node[black, above left, scale=0.8] at (1.5,5.3) {$B_{\minij}$};

\node[black, scale=1.2] at (5.6,-3.0) {$(i,j)$};
\draw[->, green, ultra thick] (5.6,-2.5)--(5.6,0.3);

\node[black, scale=1.2] at (0.5,-3.0) {$(i,\newj)$};
\draw[->, green, ultra thick] (0.5,-2.5)--(0.5,0.3);

\node[black, scale=1.5] at (-3.0,-1.0) {$A$};
\draw[->, green, ultra thick] (-2.5,-1.0)--(-1.0,-1.0);

\node[black, scale=0.0] at (11,-1.0) {$.$}; % Recenters the square by countering the $A$
\end{tikzpicture}
\end{center}
\caption{
Square contiguous submatrices anchored at $(i,j)$ 
shown when $\minij=6$.
The event $\Gamma_{i,j}$ is a function of the matrix elements in
the yellow region.
Hence these yellow entries are no longer iid when conditioned on $\Gamma_{i,j}$.
Each contiguous square submatrix anchored at $(i,j)$ is labeled inside its
top-left corner.
}
\label{Fig:Submatrices}
\end{figure}
Let $\newj =j+1-\minij$ and note that $(i,\newj)$ is the 
location in $A$ of the bottom-left corner of $B_{\minij}$.
Let $C$ be the collection of the 
left-most $\minij{-}1$ points in the bottom row of $B_{\minij}$.
Define 
$\Gamma_C = \displaystyle\bigcap_{(s,t) \in C} \theta_{s,t}$ 
to be the event that every square contiguous submatrix
anchored at a point in $C$ is nonsingular.
Then $\Gamma_{i,j} = \Gamma_C \cap \Gamma_\bottomcorner$.
If $\Gamma_{i,j}$ occurs,
then so does $\theta_{i{-}1,j{-}1}$,
so every square contiguous submatrix anchored at $(i-1,j-1)$ is nonsingular.
Then 
Lemma~\ref{lem:bad_value_condition} implies 
that the submatrices $B_1, \dots, B_{\minij}$ each have exactly 
one bad value,
namely $X(B_1), \dots, X(B_{\minij})$, respectively.

Let $U$ be the event that
the bad values of
all square contiguous submatrices of $A$ anchored at $(i,j)$ are distinct.
Since $B_1$ is a $1\times 1$ matrix, $X(B_1)=0$,
which means that
the event $U$ is equivalent to the event that
$X(B_2), \dots, X(B_{\minij})$ are all distinct and non-zero.

Then,
\begin{align}
P(U^c \mid \Gamma_{i,j})
&= P(U^c \mid \Gamma_C \cap \Gamma_\bottomcorner ) 
= \frac{P(U^c \cap \Gamma_C \mid \Gamma_\bottomcorner )}
        {P(\Gamma_C \mid \Gamma_\bottomcorner )} 
\le \frac{P( U^c \mid \Gamma_\bottomcorner )}{P(\Gamma_C \mid \Gamma_\bottomcorner )}.
 \label{eq:800}
\end{align}
We will first lower bound the denominator and then upper bound the numerator.

Conditioned on the event $\Gamma_{s,t}$ occurring,
$\theta_{s,t}$ is the event that 
$A_{s,t}$ is not a bad value for any square submatrix anchored at $(s,t)$.
Since there are at most $\min(s,t)$ bad values at $(s,t)$
and the value of the element at position $(s,t)$ is 
uniform on $\fq$ and independent of the event $\Gamma_{s,t}$,
the probability of $\theta_{s,t}^c$ given $\Gamma_{s,t}$ is
upper bounded as
\begin{align}
P(\theta_{s,t}^c \mid \Gamma_{s,t}) \le \frac{\min (s,t)}{q} \le \frac{k}{q}.
\label{eq:103}
\end{align}
Then,
\begin{align}
P(\Gamma_C \mid \Gamma_\bottomcorner )
&= \prod_{h=j'}^{j-1} 
   P\left(\theta_{i,h} \Big| \Gamma_\bottomcorner \cap 
           \bigcap_{j'\le r < h} \theta_{i,r}
    \right)\notag\\
&= \prod_{(s,t) \in C} P(\theta_{s,t} \mid \Gamma_{s,t})\label{eq:104}\\
&\ge \prod_{(s,t) \in C} \left( 1-\frac{k}{q} \right) \label{eq:101}\\
&\ge \left( 1-\frac{k}{q} \right)^k \label{eq:102} \\
&\ge \frac{1}{2} \label{eq:801}
\end{align}
where
\eqref{eq:104} follows since 
  $\Gamma_{i,j'} \subseteq \theta_{i,r}$ whenever $j' \le r < h$;
\eqref{eq:101} follows from \eqref{eq:103};
\eqref{eq:102} follows since $|C| \le k$;
and
\eqref{eq:801} follows from \eqref{eq:804}.

For each $m=2, \dots, \minij$,
let $(\bar{B}_m,u_m,v_m)$ be the corner decomposition of $B_m$.
Given $\Gamma_\bottomcorner $, 
the submatrices $\bar{B}_2, \dots, \bar{B}_{\minij}$ are nonsingular,
since their mutual bottom-right corner $(i-1,j-1)$
preceeds $(i,j')$ in raster-scan order.
Let $M$ be a $(\minij {-}1) \times (\minij {-}1)$ matrix
whose $m$th row is $\minij{-}m{-}1$ zeros
followed by $(\bar{B}_{m+1}^{-1}v_{m+1})^t$
(see Figure~\ref{Fig:M-matrix}).

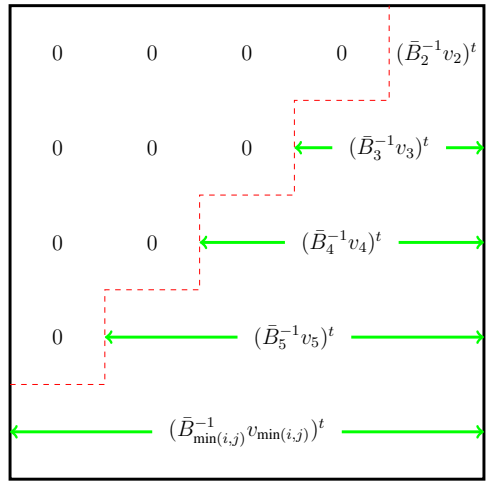
\begin{figure}[ht]
\begin{center}
\scalebox{.7}{
\begin{tikzpicture}[scale=1.8]

\draw[-, ultra thick] (0,0)--(5,0)--(5,5)--(0,5)--cycle;
\draw[dashed, red, thin] (0,1)--(1,1)--(1,2)--(2,2)--(2,3)--(3,3)--(3,4)--(4,4)--(4,5);

\node[black, scale=1] at (2.5,0.5) {$(\bar{B}_{\minij}^{-1} v_{\minij})^t$};
\node[black, scale=1] at (3.0,1.5) {$(\bar{B}_5^{-1} v_5)^t$};
\node[black, scale=1] at (3.5,2.5) {$(\bar{B}_4^{-1} v_4)^t$};
\node[black, scale=1] at (4.0,3.5) {$(\bar{B}_3^{-1} v_3)^t$};
\node[black, scale=1] at (4.5,4.5) {$(\bar{B}_2^{-1} v_2)^t$};

\draw[->, green, ultra thick] (1.5, 0.5)--(0.0,0.5);
\draw[->, green, ultra thick] (3.5, 0.5)--(5.0,0.5);

\draw[->, green, ultra thick] (2.4, 1.5)--(1.0,1.5);
\draw[->, green, ultra thick] (3.6, 1.5)--(5.0,1.5);

\draw[->, green, ultra thick] (2.9, 2.5)--(2.0,2.5);
\draw[->, green, ultra thick] (4.1, 2.5)--(5.0,2.5);

\draw[->, green, ultra thick] (3.4, 3.5)--(3.0,3.5);
\draw[->, green, ultra thick] (4.6, 3.5)--(5.0,3.5);

\node[black, scale=1] at (0.5,1.5) {$0$};

\node[black, scale=1] at (0.5,2.5) {$0$};
\node[black, scale=1] at (1.5,2.5) {$0$};

\node[black, scale=1] at (0.5,3.5) {$0$};
\node[black, scale=1] at (1.5,3.5) {$0$};
\node[black, scale=1] at (2.5,3.5) {$0$};

\node[black, scale=1] at (0.5,4.5) {$0$};
\node[black, scale=1] at (1.5,4.5) {$0$};
\node[black, scale=1] at (2.5,4.5) {$0$};
\node[black, scale=1] at (3.5,4.5) {$0$};
\end{tikzpicture}
}
\end{center}
\caption{
The matrix $M$ shown when $\minij=6$.
}
\label{Fig:M-matrix}
\end{figure}

For each $m=2, \dots, \minij$, since
$v_m = \bar{B}_m  (\bar{B}_m^{-1} v_m)$,
if the first entry of $\bar{B}_m^{-1} v_m$ is $0$, then
$v_m$ is a linear combination of the right-most $m-2$ columns of $\bar{B}_m$
(or $v_m=0$ in the case $m=2$).
But this would imply the $(m-1)\times (m-1)$ matrix anchored at $(i-1,j)$ is singular,
contradicting the occurrence of $\Gamma_\bottomcorner$.
Thus, all of the anti-diagonal elements of the matrix $M$ are non-zero,
implying that $M$ has non-zero determinant and thus has full row-rank.

Note that for each $m=2, \dots, \minij$,
the bad value $X(B_m)$ is equal to the $(m-1)$th row of $M$ multiplied by $u_{\minij}$.
In matrix form, this means
\begin{align}
M u_{\minij} &=
[ X(B_2)\ \dots\ X(B_{\minij}) ]^t
\label{eq:105}
\end{align}
and the entries of $u_{\minij}$ 
(which are the values at the positions in $C$)
are iid uniform on $\mathbb{F}_q$ 
(since only $\Gamma_\bottomcorner$ is given).
Thus,
the probability of these $\minij -1$
bad values all being distinct and non-zero given $\Gamma_\bottomcorner$ is 
\begin{align}
P( U \mid \Gamma_\bottomcorner ) 
&= \prod_{m=1}^{\minij-1} \frac{q-m}{q}\label{eq:808}\\
&\ge  \left( 1 - \frac{\minij}{q} \right)^{\minij}\notag\\
&\ge  1 - \frac{(\minij)^2}{q} \label{eq:806}
\end{align}
where
\eqref{eq:808} follows since $X(B_2), \dots, X(B_{\minij})$ are iid uniform on $\mathbb{F}_q$
from \eqref{eq:105} and Lemma~\ref{lem:uniform};
and
\eqref{eq:806} follows from Lemma~\ref{lem:Delta}
   since $\frac{\minij}{q} \le \frac{k}{q} < 1$
   for sufficiently large $q$ and $k$.

Therefore, combining 
\eqref{eq:800}, \eqref{eq:801}, and \eqref{eq:806}
gives
\begin{align*}
P(U \mid \Gamma_{i,j})
= 1 - P(U^c \mid \Gamma_{i,j})
\ge 1 - \frac{P( U^c \mid \Gamma_\bottomcorner )}{P(\Gamma_C \mid \Gamma_\bottomcorner )}
\ge 1 - \frac{2(\minij)^2}{q} .
\end{align*}
\end{proof}
% 
%---------------------------------------------------------------------------

The following lemma proves one direction of
Theorem~\ref{thm:main-contiguous}.

\begin{lemma}
For any $\lambda\in [0,\infty]$,
the probability that a random $k\times k$ matrix 
over $\fq$ is contiguous super-regular is lower bounded as
\begin{align*}
\liminf_{\substack{k,q \To \infty \\ \frac{1}{3}k^3 / q \To \lambda }}
\ProbCSR{q}{k} 
&\ge e^{-\lambda}. 
\end{align*}
\label{lem:strong-inf}
\end{lemma}

\begin{proof}
When $\lambda=\infty$ the result is trivial, so assume $\lambda<\infty$.
Let $A$ be a $k\times k$ random matrix over field $\fq$.
For each $i,j \in \{1,\dots,k\}$
and for each $m\le \minij$,
let $A_{i,j}(m)$ denote the contiguous $m\times m$ 
submatrix of $A$ anchored at $(i,j)$.
For all $i,j,u,v \in \{1, \dots, k\}$ 
we write $(u,v) \le (i,j)$ to mean that
the location $(u,v)$ comes at or before the location $(i,j)$
of the matrix $A$ in the usual raster-scan order
from left to right and top to bottom,
i.e., $(u-1)k + v \le (i-1)k + j$.

First note that $A_{i,j}(1)$ has a unique bad value of $0$ for any $(i,j)$.
Now consider $A_{i,j}(m)$ when $m \ge 2$ (and therefore $i,j \ge 2$).
Then given $\Gamma_{i,j}$ 
(which implies the event $\Gamma_{i-1,j-1}$ occurred, since $\Gamma_{i-1,j-1} \supseteq \Gamma_{i,j}$),
the upper left $(m-1) \times (m-1)$ submatrix of $A_{i,j}(m)$ is nonsingular,
and so Lemma~\ref{lem:bad_value_condition} implies that $A_{i,j}(m)$ has a unique bad value.
Therefore, 
given $\Gamma_{i,j}$,
there are at most $\minij$ bad values for position $(i,j)$ 
and so 
$P\left(\theta_{i,j}^c \mid \Gamma_{i,j}\right) \le \frac{\minij}{q}$.

The inequality $e^y \ge 1+y$ 
with $y=\frac{x}{1-x}$ implies
$1-x \ge - e^{-x/(1-x)}$ for all $x\in (0,1)$,
so whenever $k\le \minij < q$ we have:
\begin{align}
P\left(\theta_{i,j} \mid \Gamma_{i,j}\right)
\ge 1-\frac{\minij}{q}
\ge e^{-\frac{\minij/q}{1-\minij/q}}
= e^{- \frac{\minij}{q-\minij}}
\ge e^{- \frac{\minij}{q-k}}.
\label{eq:502}
\end{align}

The probability that the matrix $A$ 
is contiguous super-regular is then lower bounded as follows:
\begin{align}
\ProbCSR{q}{k} 
&= P\left(\bigcap_{(1,1) \le (i,j) \le (k,k)} \theta_{i,j}\right) \notag\\
&= \prod_{(1,1) \le (i,j) \le (k,k)} P\left(\theta_{i,j} \middle\vert 
   \bigcap_{(1,1) \le (u,v) < (i,j)} \theta_{u,v}\right) \label{eq:500}\\
&= \prod_{(1,1) \le (i,j) \le (k,k)} P\left(\theta_{i,j} \mid \Gamma_{i,j}\right) \notag\\
&\ge \prod_{(1,1) \le (i,j) \le (k,k)} e^{- \frac{\minij}{q-k}} \label{eq:501}\\
&= \exp\left\{ {- \frac{1}{q-k} \sum_{1\le i,j \le k} \minij} \right\}  \notag\\
&= e^{- \frac{\frac{1}{3} k^3 + \frac{1}{2}k^2 + \frac{1}{6}k}{q-k}} \label{eq:1001}
\end{align}
where
\eqref{eq:500} follows from the chain rule;
\eqref{eq:501} follows from \eqref{eq:502};
and
\eqref{eq:1001} follows from Lemma~\ref{lem:num_connected}.
Thus, for any $\lambda\in [0,\infty)$,
\begin{align*}
\liminf_{\substack{q,k \To \infty \\ \frac{1}{3}k^3 / q \To \lambda }} 
\ProbCSR{q}{k}
&\ge \liminf_{\substack{q,k \To \infty \\ \frac{1}{3}k^3 / q \To \lambda }} 
  e^{- \frac{\frac{1}{3} k^3 + \frac{1}{2}k^2 + \frac{1}{6}k}{q-k}}
= e^{-\lambda}.
\end{align*}
\end{proof}

%-----------------------
\begin{lemma}
For any $\lambda\in [0,\infty)$,
the probability that a random $k\times k$ matrix
over $\fq$ is contiguous super-regular is upper bounded as
\begin{align*}
\limsup_{ \substack{q,k \To \infty \\ \frac{1}{3}k^3/q \To \lambda}} \, \ProbCSR{q}{k} 
&\le e^{-\lambda}.
\end{align*}
\label{lem:Cauchy2}
\end{lemma}

\begin{proof}
Let $A$ be a random $k\times k$ matrix over $\fq$.
For each $i,j$ let $\beta_{i,j}$ be the event that there
are $\minij$ bad values at $(i,j)$.
Then $\beta_{i,j}$ occurs precisely when all possible submatrices
anchored at $(i,j)$ have unique bad values.
That is, no two such submatrices share the same bad value.

Lemma~\ref{lem:two-bad-values-strong2} implies that
for sufficiently large $k$ and $q$ and for all $i,j$,
the probability that all the bad values at $(i,j)$ are distinct,
given that all matrices anchored at locations prior 
to $(i,j)$ in raster-scan order are
are nonsingular, is lower bounded as
\begin{align}
P(\beta_{i,j} | \Gamma_{i,j})
&\ge 1 - \frac{2(\minij)^2}{q} .
\label{eq:301}
\end{align}

Since the matrix element $A_{i,j}$ is chosen uniformly at random from $\fq$,
the probability that $A_{i,j}$ is a bad value of at least one contiguous
submatrix anchored at $(i,j)$,
given that $\beta_{i,j}$ occurs,
is
\begin{align*}
P(\theta_{i,j}^c \,| \, \beta_{i,j} ) 
&= \frac{\minij}{q}.
\end{align*}
If we also know that $\Gamma_{i,j}$ occurred,
then this does not change the fact that $A_{i,j}$ is chosen uniformly
at random from $\fq$, so the same probability still holds, i.e., 
\begin{align}
P(\theta_{i,j}^c \,| \, \beta_{i,j}, \Gamma_{i,j} ) &= \frac{\minij}{q}.  \label{eq:245}
\end{align}
Thus, for all sufficiently large $k$ and $q$ and all $i,j\le k$,
\begin{align}
P(\theta_{i,j} \,| \, \Gamma_{i,j} )
&= P(\theta_{i,j} \,| \, \Gamma_{i,j}, \beta_{i,j} ) P(\beta_{i,j} | \Gamma_{i,j})
 + P(\theta_{i,j} \,| \, \Gamma_{i,j}, \beta_{i,j}^c ) P(\beta_{i,j}^c | \Gamma_{i,j})\notag\\
&= \left( 1 - \frac{\minij}{q}\right) P(\beta_{i,j} | \Gamma_{i,j})
 + P(\theta_{i,j} \,| \, \Gamma_{i,j}, \beta_{i,j}^c ) P(\beta_{i,j}^c | \Gamma_{i,j})
    \label{eq:246}\\
&\le \left( 1 - \frac{\minij}{q}\right) P(\beta_{i,j} | \Gamma_{i,j})
 + 1 - P(\beta_{i,j} | \Gamma_{i,j})\\
&=  1 - \frac{\minij}{q} P(\beta_{i,j} | \Gamma_{i,j})\notag\\
&\le 1 - \frac{\minij}{q}\left( 1 - \frac{2(\minij)^2}{q}\right) \label{eq:300}
\end{align}
where
\eqref{eq:246} follows from \eqref{eq:245};
and
\eqref{eq:300} follows from \eqref{eq:301}.
Note that 
$\frac{(\minij)^2}{q} \le \frac{k^2}{q}\To 0$ as 
$k,q\To \infty$ and $\frac{k^3}{q}\To \lambda$.
Thus, if $k$ and $q$ are sufficiently large
and $2\le i,j\le k$, then
\begin{align*}
\minij \left( 1 - \frac{2(\minij)^2}{q} \right) \ge 1
\end{align*}
in which case
\eqref{eq:300} and Lemma~\ref{lem:Delta} imply
\begin{align}
P(\theta_{i,j} \,| \, \Gamma_{i,j} )
&\le \left( 1 - \frac{1}{q}\right)^{ \minij - \frac{2(\minij)^3}{q} } 
 \le \left( 1 - \frac{1}{q}\right)^{ \minij - \frac{2k^3}{q} } .
\label{eq:238}
\end{align}
The probability that $A$ is contiguous super-regular 
is then upper bounded for sufficiently large $k$ and $q$ using
\begin{align}
\ProbCSR{q}{k}
&= P\left( \bigcap_{i,j=1}^k \theta_{i,j}\right)
 = \prod_{i,j=1}^k P(\theta_{i,j} | \Gamma_{i,j}) \notag\\
&\le \prod_{i,j=2}^k P(\theta_{i,j} | \Gamma_{i,j}) \notag\\
&\le \prod_{i,j =2}^k 
 \left(1-\frac{1}{q}
   \right)^{\minij - \frac{2k^3}{q} }
     \label{eq:236b}\\
&\le \left(1 - \frac{1}{q}\right)^{ - \frac{2k^5}{q} + \dsum_{i,j=2}^k \minij }\notag\\
&= \left(1 - \frac{1}{q}
    \right)^{ - \frac{2k^5}{q} 
          + \left( \frac{k^3}{3} + \frac{k^2}{2} + \frac{k}{6} \right)  
             -(2k-1) } \label{eq:238b}\\
&= \left(1 - \frac{1}{q}
    \right)^{q\left( - \frac{2k^5}{q^2} 
                    +  \frac{k^3}{3q} + \frac{k^2}{2q} - \frac{11k}{6q} + \frac{1}{q}
                 \right) } \notag\\
&\To e^{-\lambda}
\ \ \ \ \text{as}\ k,q \To\infty\ \text{and}\ 
                 \frac{1}{3}k^3/q \To \lambda \label{eq:243}
\end{align}
where
\eqref{eq:236b} follows from \eqref{eq:238};
\eqref{eq:238b} follows from Lemma~\ref{lem:num_connected};
and
\eqref{eq:243} follows from Lemma~\ref{lem:convergence}.
Therefore,
\begin{align*}
\limsup_{ \substack{q,k \To \infty \\ \frac{1}{3}k^3/q \To \lambda }} \, \ProbCSR{q}{k} 
&\le e^{-\lambda}.
\end{align*}
\end{proof}
%---------------------

%-----------------------
\begin{theorem} 
Let $A$ be a random $k \times k$ matrix over a field of size $q$,
where $q, k \To \infty$.\\
If $\frac{k^3/3}{q} \To \lambda \in [0,\infty]$, 
then 
$P(A \text{\ is contiguous super-regular}) 
   \To e^{-\lambda}$.
\label{thm:main-contiguous}
\end{theorem}
\begin{proof}
For all $\lambda\in [0,\infty)$,  
Lemma~\ref{lem:strong-inf} and
Lemma~\ref{lem:Cauchy2} imply
\begin{align}
\limsup_{ \substack{q,k \To \infty \\ \frac{1}{3}k^3/q \To \lambda}} \, \ProbCSR{q}{k} 
&\le e^{-\lambda} 
\le \liminf_{ \substack{q,k \To \infty \\ \frac{1}{3}k^3/q \To \lambda}} \ \ProbCSR{q}{k} 
\label{eq:final-strong-result}
\end{align}
so $\ProbCSR{q}{k} \To e^{-\lambda}$
as $q,k \To \infty$ and $\frac{1}{3}k^3/q \To \lambda$.

When $\lambda=\infty$,
Lemma~\ref{lem:strong-inf} also implies the
right-hand bound in \eqref{eq:final-strong-result}.

Now assume $\lambda=\infty$
and let us show the left-hand bound in \eqref{eq:final-strong-result}.
Note that $\ProbCSR{q}{k}$ is monotone decreasing in $k$ 
since any square matrix larger than
$k\times k$ contains multiple $k\times k$ contiguous submatrices.
For all positive real numbers 
$\lambda'$ and $\epsilon>0$,
and for sufficiently large $k$ and $q$,
we have $\frac{1}{3}k^3/q > \lambda'$,
and thus $\ProbCSR{q}{k} \le e^{-\lambda'} + \epsilon$
by Lemma~\ref{lem:Cauchy2}.
Taking $\lambda'$ arbitrarily large and $\epsilon$ arbitrarily small
implies $\ProbCSR{q}{k} \To 0$.
\end{proof}

We note that Theorem~\ref{thm:main-contiguous} still holds even when
$k$ is bounded 
(i.e., when $\lambda=0$),
because Lemma~\ref{lem:strong-inf} can be seen to hold 
without requiring $k\To\infty$.

%================================================
\clearpage
\section{Enumeration of small super-regular matrices}
\label{sec:small-matrices}

The numbers of $1\times 1$ and $2\times 2$ super-regular and
contiguous super-regular matrices are easily seen to be
the polynomials
$\CountSR{q}{1}=\CountCSR{q}{1}=q-1$
and
$\CountSR{q}{2}=\CountCSR{q}{2}=(q-1)^3(q-2)$.

In this section we focus on how many 
$3\times 3$ and $4\times 4$
matrices are super-regular or contiguous super-regular
and we include some computational results.

% ---------------------------------------------------------------------------
We say a matrix over a field is in 
\textit{normal form} if its first column and first row are all $1$s.
For each $k\ge 1$, 
define a binary relation $\sim$ on $k\times k$ matrices
over $\fq\setminus \{0\}$
such that for any matrices $A$ and $B$,
we write $A\sim B$ if $B$ can be obtained by multiplying some 
(possibly empty) 
sequence of rows and columns of $A$
by non-zero field elements.
One can verify that  $\sim$
is an equivalence relation and 
each equivalence class has exactly one matrix in normal form.

The following lemma shows that in order to count the number of 
super-regular 
(or contiguous super-regular) 
square matrices over a given field
it suffices to assume each matrix is in normal form,
thus reducing the computational complexity.

\begin{lemma}
The number 
$\CountSR{q}{k}$ 
(respectively, $\CountCSR{q}{k}$) 
of $k\times k$ 
super-regular 
(respectively, contiguous super-regular) 
matrices over $\fq$
is $(q-1)^{2k-1}$ times the number 
of $k\times k$ 
super-regular 
(respectively, contiguous super-regular) 
matrices over $\fq$ in normal form.
\label{ones-in-top-row-left-col}
\end{lemma}

\begin{proof}
Each matrix in an equivalence class under $\sim$ can be obtained from 
the class's unique normal form matrix by
first multiplying the matrix's $k$ rows by nonzero field elements
and then multiplying the matrix's $k-1$ right-most columns by nonzero field elements.
The $2k-1$ multipliers used in this process uniquely determine the resulting matrix.
Thus each equivalence class contains exactly $(q-1)^{2k-1}$ matrices.
Multiplying a row or column of a matrix by a nonzero constant does not change whether
or not the matrix is super-regular or contiguous super-regular, 
so for each equivalence class,
either all or none of its matrices are super-regular.
\end{proof}

% ---------------------------------------------------------------------------
% ---------------------------------------------------------------------------
\begin{lemma}
If $k \ge 3$, 
then there are no $k\times k$ 
super-regular nor contiguous super-regular matrices
over $\field{2}$ or $\field{3}$,
i.e.,
$\CountSR{2}{k}=\CountCSR{2}{k}=\CountSR{3}{k}=\CountCSR{3}{k}=0$.
\label{lem:field-sizes-2-3}
\end{lemma}
% ---------------------------------------------------------------------------

\begin{proof}
If $q=2$, then any contiguous super-regular matrix would consist of all $1$s, 
which itself is singular, a contradiction.
So suppose $q=3$ and assume a square contiguous super-regular matrix is in normal form. 
To avoid contiguous singular $2\times 2$ submatrices,
the second row must start $121$ and therefore the third row must start $111$
which means the top-left $3\times 3$ contiguous submatrix is singular since its first
and third rows are both $111$.
Thus, the matrix is in fact not contiguous super-regular.
Since every super-regular matrix is also contiguous super-regular,
there are no super-regular matrices 
when $q=2$ and $k\ge 2$,
or 
else when $q=3$ and $k\ge 3$.
\end{proof}

\begin{lemma}
For any $k$,
the probability $\ProbSR{q}{k}$ 
(respectively, $\ProbCSR{q}{k}$)
that a random $k\times k$ matrix over $\fq$ is 
(respectively, contiguous) super-regular
tends to $1$ as $q\To\infty$.
\label{lem:super-regular-as-q-grows}
\end{lemma}

\begin{proof}
Let $M$ be a random $k\times k$ matrix over $\fq$ and for each $j$
let $S_j$ denote the set of all $j\times j$ submatrices of $M$.
Then
\begin{align*}
%P(A\text{ is super-regular})
\ProbCSR{q}{k}
&\ge \ProbSR{q}{k}\\
&= 
P\left( \bigcap_{j=1}^k \bigcap_{U\in S_j}   \{ U \text{ is nonsingular} \} \right)\\
&\ge 1 - \sum_{j=1}^k \sum_{U\in S_j}  P \{ U \text{ is singular} \} \\
&= 1 - \sum_{j=1}^k \binom{k}{j}^2 \left( 1 - \prod_{i=1}^j (1-q^{-i}) \right)\\
&\ge 1 - \binom{2k}{k}\left( 1 - (1-q^{-1})^k  \right) \\
&\To 1 \ \ \ \text{as}\ q\To\infty .
\end{align*}
\end{proof}

\subsection{Counting $3\times 3$ super-regular matrices}
\label{sec:3x3}

This section analyzes the counts
$\CountSR{q}{3}$ and $\CountCSR{q}{3}$.

\begin{theorem}
The number of $\ 3 \times 3$ matrices over a finite field $\fq$
that are
  \begin{itemize}
    \item[(a)] super-regular is $\CountSR{q}{3} = (q-1)^5 (q-2)(q-3)(q^2-9q+21)$, and
    \item[(b)] contiguous super-regular is $\CountCSR{q}{3} = (q-1)^5 (q-2)(q-3)(q^2-4q+5)$.
  \end{itemize}
  \label{thm:3x3}
\end{theorem}

\begin{proof}
Part (a) is due to
Skorobogatov~\cite[Proposition 3.2]{Skorobogatov-1992}.

For Part (b),
let $M$ be a $3 \times 3$ matrix over $\fq$ in normal form, such that
\begin{align*}
  M &= \left[\begin{array}{ccc}
              1 & 1 & 1 \\
              1 & A & B \\
              1 & C & D 
        \end{array} \right] .
\end{align*}
It can easily be verified that $M$ is contiguous super-regular if and only if
\begin{align}
  A,B,C,D  &\ne 0  \label{eq:3s_1} \\
  A - 1    &\ne 0 \label{eq:3s_2} \\
  B - A  &\ne 0 \label{eq:3s_4} \\
  C - A  &\ne 0 \label{eq:3s_6} \\
  A D - B C &\ne 0 \label{eq:3s_10} \\
  (B - A) - C (B - 1) + D (A - 1) &\ne 0 \label{eq:3s_11}.
\end{align}

The element $A$ may be any non-zero element of $\fq$ except
\begin{align}
A &\ne 1 \label{eq:3s_13}
\end{align}
since this choice satisfies \eqref{eq:3s_1} and \eqref{eq:3s_2}
and does not violate any of \eqref{eq:3s_4} -- \eqref{eq:3s_11}.
There are $(q-2)$ such choices of $A$.

Given the value of $A$, 
the element $B$ may be any non-zero element of $\fq$ except
\begin{align}
B &\ne A \label{eq:3s_14}
\end{align}
since this choice satisfies \eqref{eq:3s_1}, \eqref{eq:3s_2}, and \eqref{eq:3s_4}
and does not violate \eqref{eq:3s_6} -- \eqref{eq:3s_11}.
There are $(q-2)$ such choices for $B$.

Given the values of $A, B$, 
the element $C$ may be any non-zero element of $\fq$ except
\begin{align}
C &\ne A  \label{eq:3s_15}
\end{align}
since this choice satisfies \eqref{eq:3s_1} -- \eqref{eq:3s_6}
and does not violate \eqref{eq:3s_10} or \eqref{eq:3s_11}.
There are $(q-2)$ such choices for $C$.

By \eqref{eq:3s_13}, $A - 1 \ne 0$,
so let
$$e = (A - 1)^{-1} (C (B - 1) - (B - A)).$$
Given the values of $A,B,C$, 
the element $D$ may be any non-zero element of $\fq$ except
\begin{align*}
D &\ne A^{-1} B C\\
D &\ne e
\end{align*}
since these choices satisfy \eqref{eq:3s_1} -- \eqref{eq:3s_11}.
In fact, 
these conditions on $A, B, C, D$ are necessary and sufficient 
for $M$ being contiguous super-regular.
However, 
$0$, $e$, and $A^{-1} B C$ 
may not be distinct for all choices of $A, B, C$.

Clearly $A^{-1} B C \ne 0$.
We also have
\begin{align}
e = A^{-1} B C 
& \; \Longleftrightarrow \;
  (C - A) (B - A) = 0 \label{eq:3s_18} \\
  e = 0
& \; \Longleftrightarrow \;
  C (B - 1) = (B - A) \label{eq:3s_19} .
\end{align}

Then, \eqref{eq:3s_14}, \eqref{eq:3s_15}, and \eqref{eq:3s_18}
imply $e \ne A^{-1} B C.$

Now suppose $B = 1$.
We note that \eqref{eq:3s_13} implies $1 \ne A$,
so this is a valid choice of $B$.
Then \eqref{eq:3s_14} and \eqref{eq:3s_19} imply $e \ne 0$.
So when $B = 1$,
there are $(q-3)$ distinct choices for $D$,
namely any field element other than 
$0$, $e$, or $A^{-1}BC$,
and there are $(q-2)$ distinct choices for $C$,
namely any field element other than 
$0$ or $A$.

Now suppose $B$ is any non-zero element of $\fq$ except
$$A \ \ \text{ or } \ \ 1. $$
Then \eqref{eq:3s_19} becomes
$$e= 0 \; \Longleftrightarrow \; C = (B - A) (B - 1)^{-1} .$$
We also have
$$(B - A) (B - 1)^{-1} \ne 0 $$
and 
$$(B - A) (B - 1)^{-1} =  A \; \Longleftrightarrow \; A =  1 $$
which together with \eqref{eq:3s_13} implies
$(B - A) (B - 1)^{-1} $
is a valid choice of $C$.

Hence, when 
$$B \not \in \{0, \; 1, \;  A\}
\; \text{ and } \;
C = (B - A) (B - 1)^{-1} $$
there are $(q-2)$ choices for $D$,
namely any field element other than 
$0$ or $A^{-1}BC$.
Also, when 
$$B \not \in \{0, \; 1, \;  A\}
\; \text{ and } \;
C \not \in \{0, \; A, \; 
  (B - A) (B - 1)^{-1} \}$$
there are $(q-3)$ choices for $D$,
namely any field element other than 
$0$, $e$, or $A^{-1}BC$.

Thus, given a valid choice for $A$,
there are
$$(1) (q-2) (q-3) + (q-3) ((1) (q-2) + (q-3)(q-3) ) = (q-3) (q^2 - 4 q + 5)$$
choices for $B, C, D$,
which along with 
with Lemma~\ref{ones-in-top-row-left-col}
and the fact that there are $(q-2)$ choices for $A$
proves the result.
\end{proof}

%================================================

\subsection{Counting $4\times 4$ super-regular matrices}
\label{sec:4x4}

This section analyzes the counts 
$\CountSR{q}{4}$ and $\CountCSR{q}{4}$.

Section~\ref{sec:3x3}
showed that the number of super-regular and
the number of contiguous super-regular
$3\times 3$ matrices 
are both polynomial functions of the field size $q$.

For $4\times 4$ matrices,
analogous results are not known and do not appear straightforward,
but we do provide some insight.

First, using a computer we computationally determine values of $\CountSR{q}{4}$ 
for field sizes $q$ up to $71$ (and also equal to $83$ and $97$).
Then we use these observations to prove
that the number of 
super-regular $4\times 4$ matrices over $\fq$
is not a quasi-polynomial function with period less than $7$.
A consequence is that $\CountSR{q}{4}$ 
is not a polynomial function of the field size $q$.

Finally, for contiguous super-regular $4\times 4$ matrices
we numerically calculate the exact value of $\CountCSR{q}{4}$ 
for all finite fields with $q\le 67$
and then determine a low-degree polynomial function of the field size
which agrees with all of our computer-calculated matrix counts.
Of the $26$ counts we obtained by computer calculations for $4 \le q \le 67$,
only $7$ are needed to determine such a polynomial,
and yet all $26$ exactly match the values predicted by the polynomial,
thus providing significant confidence as to its correctness.
However, we do not have a proof of its correctness, 
but rather conjecture it based on the extensive experimental evidence.

Lemma~\ref{lem:field-sizes-2-3} showed that there are no
super-regular nor contiguous super-regular matrices over fields of size $2$ or $3$.
The following lemma extends this result to super-regular matrices over
fields of size $4$ or $5$.

\begin{lemma}
There are no $4\times 4$ 
super-regular matrices over fields of size $4$ or $5$,
i.e., $\CountSR{4}{4} = \CountSR{5}{4} = 0$.
\label{lem:field-sizes-4-5}
\end{lemma}

\begin{proof}
By Lemma~\ref{ones-in-top-row-left-col},
it suffices to show the result
for a generic $4\times 4$ matrix in normal form,
i.e.,
\begin{align*}
M &=
\begin{bmatrix}
1 & 1 & 1 & 1 \\
1 & A & B & C \\
1 & D & E & F \\
1 & G & H & I 
\end{bmatrix} .
\end{align*}

\begin{itemize}
\item $q=4$\\
The nonzero elements $A$, $B$, $C$ must be distinct and not equal to $1$ 
in order to avoid a
singular $2\times 2$ submatrix with the first row.
Since two of the four field elements must be avoided,
$M$ cannot be super-regular and therefore
$\CountSR{4}{4} = 0$.

\item $q=5$\\
Since the top row  of $M$ is all $1$s,
no other row can have $2$ identical elements, 
or else $M$ would have a singular $2\times 2$ submatrix. 
So the rows of the bottom-right
$3 \times 3$ submatrix of $M$ are all permutations of the $3$ elements
of $GF(5)\setminus \{0,1\}$.
Without loss of generality we can take 
$A=2$, $B=3$, $C=4$, $D=3$, and $G=4$.
But then we must have $E=4$ and $F=2$ 
to avoid a singular $2\times 2$ submatrix,
and similarly $H=2$ and $I=3$.
But then the $2\times 2$ submatrix 
$\begin{bmatrix} 
B & C \\
E & F
\end{bmatrix}$
has determinant 
$BF - EC = (3\cdot 2) - (4\cdot 4) = 0 \mod 5$,
so $M$ is not super-regular.
Thus $\CountSR{5}{4} = 0$.
\end{itemize}
\end{proof}

\subsubsection{Computer enumeration}
\label{sec:computer-simulation}

We wrote parallelized C++ software%
\footnote{
The software used is posted at
the repository 
https://github.com/kzeger/Super-Regular-Matrix-Counter.git .
}
to count by exhaustive search
the number of super-regular
and contiguous super-regular $4\times 4$ matrices over $\fq$ in normal form, i.e.,
\begin{align*}
\begin{bmatrix}
1 & 1 & 1 & 1\\
1 & A & B & C\\
1 & D & E & F\\
1 & G & H & I
\end{bmatrix}.
\end{align*}
The code is used
to obtain the data in Table~\ref{tab:data},
which in turn is used in the proofs of
Theorem~\ref{thm:not-quasi} and
Theorem~\ref{thm:4x4-poly},
as well as in formulating
Conjecture~\ref{conj:poly}.

Note that a matrix is contiguous super-regular 
if and only if 
its transpose is contiguous super-regular.
To exploit this fact for computational complexity reduction, 
our code assumes
$B \le D$ and then adds two to the count
$\frac{\CountCSR{q}{4}}{(q-1)^7}$
every time a contiguous super-regular matrix is found with $B<D$,
and adds one to the count if $B=D$.

Also note that a matrix is super-regular and in normal form
if and only if
any permutation of its rows and columns, 
other than the first row and first column,
results in a super-regular matrix in normal form.
To exploit this fact our code assumes
that $A<B<C$ and $D<G$,
thus reducing the computational complexity by a factor of $12$.
In this case, 
the code adds $12$ to the count
$\frac{\CountSR{q}{4}}{(q-1)^7}$
every time a super-regular matrix is found.
Table~\ref{tab:data}
reports the values computed for 
$\CountSR{q}{4}(q-1)^{-7}$ and $\CountCSR{q}{4}(q-1)^{-7}$.

The code was optimized to reduce complexity for
parallel processing with dynamic scheduling.
The maximum time consuming computation was for $q=97$,
whose resulting count $\frac{\CountSR{q}{4}}{(q-1)^7}$ 
was nearly $400$ quadrillion,
and which ran for $56$ hours on
a SDSC Expanse supercomputer using
64 CPUs (AMD EPYC 7742),
each with 64 cores
% (* 64 64 56) 229376
for a total of $229,376$ core-hours.
\begin{table}
\renewcommand{\arraystretch}{1.3} % Add space between rows
\begin{center}
\scalebox{0.9}{
\begin{tabular}{|r|r|r|}
\hline
$q$
& Non-contiguous
& Contiguous\\
& $\CountSR{q}{4}(q-1)^{-7}$ & $\CountCSR{q}{4}(q-1)^{-7}$ \\ \hline\hline
 2 &                       0 &                         0 \\ \hline
 3 &                       0 &                         0 \\ \hline
 4 &                       0 &                        58 \\ \hline
 5 &                       0 &                     4,500 \\ \hline
 7 &                     120 &                   780,640 \\ \hline
 8 &                     720 &                 4,650,030 \\ \hline
 9 &                  36,360 &                20,667,108 \\ \hline
11 &                 626,544 &               228,641,184 \\ \hline
13 &              14,503,440 &             1,525,744,660 \\ \hline
16 &             464,227,344 &            14,631,667,414 \\ \hline
17 &           1,165,633,560 &            27,838,474,020 \\ \hline
19 &           5,801,345,760 &            89,105,123,008 \\ \hline
23 &          74,373,486,840 &           629,431,074,720 \\ \hline
25 &         212,300,100,840 &         1,455,886,562,020 \\ \hline
27 &         543,399,479,280 &         3,135,372,175,200 \\ \hline
29 &       1,272,945,018,960 &         6,357,690,556,308 \\ \hline
31 &       2,770,956,779,304 &        12,245,127,752,704 \\ \hline
32 &       3,992,200,415,280 &        16,706,559,415,590 \\ \hline
37 &      20,421,544,260,000 &        68,460,330,412,180 \\ \hline
41 &      62,690,828,201,064 &       183,909,349,902,564 \\ \hline
43 &     104,655,329,386,320 &       290,219,735,307,040 \\ \hline
47 &     269,248,344,742,440 &       677,931,197,232,960 \\ \hline
49 &     431,866,172,744,112 &     1,007,310,765,667,108 \\ \hline
53 &     943,630,646,029,920 &     2,118,192,838,017,300 \\ \hline
59 &   2,836,518,071,582,160 &     5,822,477,914,000,608 \\ \hline
61 &   3,980,663,747,410,224 &     7,964,743,058,190,484 \\ \hline
64 &   6,468,712,839,349,200 &    12,496,281,215,714,758 \\ \hline
67 &  10,253,864,595,564,480 &    19,189,657,322,417,920 \\ \hline  
71 &  18,308,524,066,455,384 &                           \\ \hline
83 &  85,743,863,225,899,200 &                           \\ \hline
97 & 391,851,063,583,777,560 &                           \\ \hline
%103&                     TBD &                      \\ \hline
%107&                     TBD &                      \\ \hline
\end{tabular}
}
\end{center}
\caption{
The table entries show the number of super-regular 
and contiguous super-regular $4\times 4$ matrices over $\fq$
divided by $(q-1)^7$,
as computed by exhaustive computer search.
}
\label{tab:data}
\end{table}

\subsubsection{Non-contiguous submatrices}

By Lemma~\ref{lem:MDS-conditions},
the number of super-regular $k\times k$ matrices over $\fq$
is equal to the number of $k\times 2k$ MDS matrices over $\fq$,
i.e., $\CountSR{q}{k} = \CountMDS{q}{2k}{k}$.
For $k=3$ this quantity is given in Theorem~\ref{thm:3x3}(a),
but it remains unknown when $k=4$
(e.g., see \cite{Isham-2022}). % See Section 3,3 page 26 of ArXiv version
Our Table~\ref{tab:data}
gives the values of
$\CountSR{q}{k}$
for $k=4$ for all field sizes up to and including $q=71$,
and also for $q=83, 97$.

A function $f$ on $\Z$ is a \textit{quasi-polynomial of period $m$} if
there exist $m$ polynomials $g_0, \dots, g_{m-1}$ satisfying
$f(q) = g_i(q)$ whenever $q = i\mod m$.

It is an open question whether $\CountSR{q}{4}$ is a quasi-polynomial,
but we show in 
Theorem~\ref{thm:not-quasi} 
that it cannot be a quasi-polynomial of period less than $7$,
and therefore it cannot be a polynomial either
(where the domain is restricted to powers of primes).

Lagrange's interpolation formula fits a unique polynomial of degree $n-1$
to $n$ data points in the plane.
Equivalently, a unique monic polynomial of degree $n$ fits the same data,
as stated in the following lemma.
\begin{lemma}
If 
$a_1, \dots, a_n$ 
are distinct reals and
$b_1, \dots, b_n$ 
are reals,
then 
\begin{align*}
f(x) &= x^n + \sum_{i=1}^n 
        (b_i - a_i^n) \prod_{\substack{ 1\le j \le n \\ j\ne i}}
        \frac{x-a_j}{a_i-a_j}
\end{align*}
is the unique monic polynomial of degree $n$ satisfying $f(a_i)=b_i$ for all $i=1, \dots, n$.
\label{lem:Lagrange-interpolation}
\end{lemma}

Note that in Lemma~\ref{lem:Lagrange-interpolation},
if the numbers $a_1, \dots, a_n$ and $b_1, \dots, b_n$ are all integers,
then the function $f$ generally has rational coefficients.

\clearpage

\begin{theorem}
If $m\le 6$, then
$\CountSR{q}{4}$ is not a quasi-polynomial of period $m$.
\label{thm:not-quasi}
\end{theorem}

\begin{proof}
If a function is a quasi-polynomial of some period,
then it is also a quasi-polynomial of any integer multiple of that period,
so it suffices to show that 
$\CountSR{q}{4}$ is not a quasi-polynomial of any period in $\{4,5,6\}$.

Suppose to the contrary that $\CountSR{q}{4}$ is a quasi-polynomial 
of some period $m\in \{4,5,6\}$.
Then for each $i\in \{0, \dots, m-1\}$
there exist polynomials $h_{m,i}$ such that
$\CountSR{q}{4} = h_{m,i}(q)$ whenever $q=i\mod m$.
There are at most $(q-1)^{16}$ matrices of size $4\times 4$ over $\fq\backslash\{0\}$,
so $\CountSR{q}{4}\le (q-1)^{16}$ and therefore $\text{deg}(h_{m,i}) \le 16$ for each $i$.
Lemma~\ref{lem:super-regular-as-q-grows}
implies that $\displaystyle\lim_{q\To\infty} \CountSR{q}{4} q^{-16} = 1$
so each $h_{m,i}$ is monic and of degree $16$.
Furthermore,
since each $h_{m,i}$ maps integers to integers,
Lemma~\ref{lem:Lagrange-interpolation} implies that
$h_{m,i}\in\Q[q]$.

By Euclidean division, there exist polynomials 
$s_{m,i}, r_{m,i}\in \Q[q]$ such that
\begin{align*}
h_{m,i}(q) = (q-1)^7 s_{m,i}(q) + r_{m,i}(q)
\end{align*}
and $\text{deg}(r_{m,i})\le 6$.
There exists a large enough integer $A_{m,i}$ such that $A_{m,i}s_{m,i}\in \Z[q]$.
Thus, $A_{m,i}s_{m,i}(q)\in \Z$ for all $q\in\Z$.
If $q=i\mod m$, then
\begin{align*}
\frac{\CountSR{q}{4}}{(q-1)^7}
&= \frac{h_{m,i}(q)}{(q-1)^7} 
= s_{m,i}(q) + \frac{r_{m,i}(q)}{(q-1)^7}
\end{align*}
which is guaranteed by
Lemma~\ref{ones-in-top-row-left-col} to be an integer
(i.e. the number of $4\times 4$ super-regular matrices on $\fq$ in normal form).
This implies that $A_{m,i}\cdot\frac{r_{m,i}(q)}{(q-1)^7}\in\Z$ for all $q\in\Z$.
But as $q\To\infty$, we have $\frac{r_{m,i}(q)}{(q-1)^7}\To 0$ so eventually
$\left| A_{m,i}\cdot \frac{r_{m,i}(q)}{(q-1)^7} \right|<1$,
and the only integer value this quantity can be is zero.
Therefore $r_{m,i}(q)=0$ for all $q$,
which implies 
$h_{m,i}(q) = (q-1)^7 s_{m,i}(q)$
and in fact $s_{m,i}(q)\in \Z$ for all $q$.

We will next show that the data computed in Table~\ref{tab:data} 
contradicts this conclusion for each of the three cases:
$(m,i) \in \{ (4,3), (5,2), (6,5)\}$.
The field sizes in Table~\ref{tab:data} have the following properties:
\begin{itemize} 
\item $3,7,11,19,23,27,31,43,47 = 3\mod 4$ % Also 59 but we don't need it.
\item $2,7,17,27,32,37,47,67,97 = 2\mod 5$ 
\item $5,11,17,23,29,41,47,53,59,71 = 5\mod 6$ .
\end{itemize}

Each of the three polynomials
$s_{4,3}(q)=\frac{h_{4,3}(q)}{(q-1)^7}$,
$s_{5,2}(q)=\frac{h_{5,2}(q)}{(q-1)^7}$, and
$s_{6,5}(q)=\frac{h_{6,5}(q)}{(q-1)^7}$
is monic, has degree $9$, and is integer valued.
We use Lemma~\ref{lem:Lagrange-interpolation}
to construct the unique polynomial of degree $9$ for the above three cases
using the lowest $9$ field sizes of the correct residue modulo the period
and then derive a contradiction.
For $m=6$ the contradiction is obtained by evaluating the polynomial at a $10$th
data value in the table (i.e., $q=71$),
whereas for $m\in \{4,5\}$ we only computed $9$ values, and then used
$59 = 3\mod 4$ and $107 = 2\mod 5$, respectively, as the inputs
and obtain non-integer output values from the polynomials.
\begin{itemize} 

\item $m=4$:
Using the field sizes
$q=3,7,11,19,23,27,31,43,47$,
Lemma~\ref{lem:Lagrange-interpolation}
gives
\begin{align*}
\frac{h_{4,3}(q)}{(q-1)^7}
&= q^9 
- \frac{2681744467}{43253760} q^8 
+ \frac{9245759107}{5406720} q^7 
- \frac{297184155209}{10813440} q^6 \\
&\ \ \ + \frac{1529366777377}{5406720} q^5 
- \frac{41672290017121}{21626880} q^4 
+ \frac{46784697629773}{5406720} q^3 \\
&\ \ \ - \frac{24187428071411}{983040} q^2 
+ \frac{72190423190261}{1802240} q 
- \frac{37139179301667}{1310720} \\
\therefore 
\frac{h_{4,3}(59)}{(59-1)^7} &=  \frac{31201695993215664}{11}
\not\in\Z.
\end{align*}
\item $m=5$: 
Using the field sizes
$q=2,7,17,27,32,37,47,67,97$,
Lemma~\ref{lem:Lagrange-interpolation}
gives
\begin{align*}
\frac{h_{5,2}(q)}{(q-1)^7}
&= q^9 
- \frac{6029902971623}{97256250000} q^8 
+ \frac{166312835715263}{97256250000} q^7 
- \frac{55688739894322}{2026171875} q^6 \\
&\ \ \ + \frac{352934580939503}{1246875000} q^5 
- \frac{965086540399873}{498750000} q^4 
+ \frac{22007165789434769}{2493750000} q^3 \\
&\ \ \ - \frac{1282531884642030203}{48628125000}  q^2 
+ \frac{578672659199202089}{12157031250} q 
- \frac{3554877260305056}{96484375} \\
\therefore 
\frac{h_{5,2}(107)}{(107-1)^7} &=  \frac{249533501634221249520}{247}
\not\in\Z.
\end{align*}
\item $m=6$: 
Using the field sizes
$q=5,11,17,23,29,41,47,53,59$,
Lemma~\ref{lem:Lagrange-interpolation}
gives
\begin{align*}
\frac{h_{6,5}(q)}{(q-1)^7}
&= q^9 
- \frac{1874451481}{30233088} q^8 
+ \frac{6462323435}{3779136} q^7 
- \frac{207693950953}{7558272} q^6 \\
&\ \ \ + \frac{1068386076455}{3779136} q^5 
- \frac{29071357368383}{15116544} q^4 
+ \frac{32508639629681}{3779136} q^3 \\
&\ \ \ - \frac{183054971721325}{7558272} q^2 
+ \frac{145884846567485}{3779136} q 
- \frac{793913123691025}{30233088} \\
\therefore 
\frac{h_{6,5}(71)}{(71-1)^7} &= 18308524066623384
\ne \frac{\CountSR{71}{4}}{(71-1)^7} = 18308524066455384 .
% These differ by exactly 168,000 ... weird
\end{align*}
\end{itemize} 
Thus, we have shown that 
$\frac{h_{4,3}(q)}{(q-1)^7}$,
$\frac{h_{5,2}(q)}{(q-1)^7}$, and
$\frac{h_{6,5}(q)}{(q-1)^7}$
are not the required integer-valued polynomials, a contradiction.
\end{proof}

The following corollary follows immediately from
Theorem~\ref{thm:not-quasi}
using a period of $m=1$.
This corollary could be proven directly in a simpler way
by noting that if $\CountSR{q}{4}$ were a polynomial then
it would be divisble by $(q-2)(q-3)(q-4)(q-5)$
since
$\CountSR{2}{4} = \CountSR{3}{4} = \CountSR{4}{4} = \CountSR{5}{4} = 0$
and therefore a polynomial of only degree $5$ would need to be constructed.

\begin{corollary}
The number $\CountSR{q}{4}$ of $4\times 4$ 
super-regular matrices over $\fq$ is
not a polynomial.
\label{cor:not-a-polynomial}
\end{corollary}

%================================================

\subsubsection{Contiguous submatrices}

Define the polynomial:
\begin{align*}
f(q) &= q^7 - 18q^6 + 143q^5 - 654q^4 + 1874q^3 - 3400q^2 + 3671q - 1855.
\end{align*}

\begin{theorem}
If the number $\CountCSR{q}{4}$
of contiguous $4\times 4$ 
super-regular matrices over $\fq$ is a polynomial,
then 
\begin{align*}
\CountCSR{q}{4}
&= (q-1)^7 (q-2) (q-3) f(q).
\end{align*}
\label{thm:4x4-poly}
\end{theorem}

\begin{proof}
Suppose $\CountCSR{q}{4}$ is a polynomial.
We know that $(q-1)^7$ divides $\CountCSR{q}{4}$
and Lemma~\ref{lem:field-sizes-2-3} implies
$\CountCSR{2}{4} = \CountCSR{3}{4} = 0$ 
so $q-2$ and $q-3$ also divide the polynomial.
Lemma~\ref{lem:super-regular-as-q-grows}
implies that $\CountCSR{q}{4} q^{-16} \To 1$ as $q\To\infty$,
so $\CountCSR{q}{4}$ must be a monic polynomial of degree $16$.
Thus, $\CountCSR{q}{4}$ equals $(q-1)^7(q-2)(q-3)$ times
a monic polynomial of degree $7$.

By Lemma~\ref{lem:Lagrange-interpolation},
any $7$ points on this monic polynomial of degree $7$ uniquely determine it.
We used the computer-generated values of
$\frac{\CountCSR{q}{4}}{(q-1)^7}$
from Table~\ref{tab:data}
for $q=4,5,7,8,9,11,13$ to determine that the 
$7$th-degree polynomial is indeed $f(q)$.
\end{proof}

We also verified that each of the other $19$ 
values of $\frac{\CountCSR{q}{4}}{(q-1)^7}$ in Table~\ref{tab:data}
(i.e., when $16 \le q \le 67$)
equals $(q-2)(q-3)f(q)$,
which would seem rather unlikely unless 
$\CountCSR{q}{4}$ is indeed a polynomial.
This motivates the following conjecture.

\begin{conjecture}
The number $\CountCSR{q}{4}$ of contiguous $4\times 4$ 
super-regular matrices over $\fq$ is 
\begin{align*}
\CountCSR{q}{4}
&= (q-1)^7 (q-2) (q-3) f(q).
\end{align*}
\label{conj:poly}
\end{conjecture}

% ---------------------------------------------------------------------------
\clearpage

\section{Experimental evidence and asymptotic convergence rates}
\label{sec:Conjecture}

Theorem~\ref{thm:main-contiguous} demonstrates that 
for random $k\times k$ matrices over $\fq$,
if $\frac{1}{3}k^3q^{-1}\To \lambda\in [0,\infty]$,
then the probability that all contiguous submatrices are nonsingular
converges as $\ProbCSR{q}{k} \To e^{-\lambda}$.
Corollary~\ref{cor:main-superregular} demonstrates that
for the non-contiguous case,
if $4^k k^{-1/2} q^{-1} \To \lambda\in \{0,\infty\}$,
then the probability that all submatrices are nonsingular
converges as  $\ProbSR{q}{k} \To e^{-\lambda}$.

Other than the special cases 
when $\lambda$ equals either $0$ or $\infty$,
we do not presently know if 
$\ProbSR{q}{k} \To e^{-\lambda}$ for any $\lambda\in (0,\infty)$,
but it seems plausible to us,
and we conjecture it to be true for all such $\lambda$.
More generally, 
we conjecture Theorem~\ref{thm:main-MDS}(c) to be true even without 
the requirement $k/n\to 0$.
\begin{conjecture}
Let $C$ be a linear $[n,k]$ code chosen uniformly at random
over a field of size $q$,
where $n,q \To\infty$.
If $\frac{1}{q}\binom{n}{k}\To \lambda\in (0,\infty)$,
then $P(C \text{ is MDS}) \To e^{-\lambda}$.
\label{conj:converge}
\end{conjecture}

We do not see a straightforward way to prove this using the 
Chen-Stein method that we exploited in the 
present proof of Theorem~\ref{thm:main-MDS}(c),
nor using the technique in the proof of Theorem~\ref{thm:main-contiguous}.
We thus leave this conjecture as an open question,
but we do demonstrate some numerical results (when $n=2k$)
which support the conjecture
and demonstrate convergence rates.

Figure~\ref{fig:conjecture} plots 
(as blue dots)
numerically computed 
(from $1000$ randomly chosen matrices)
estimates of 
$\ProbSR{q}{10}$ and $\ProbCSR{q}{10}$
versus 
$\lambda = \binom{2k}{k}q^{-1}$ (for super-regular)
and versus
$\lambda = \frac{1}{3}k^3 q^{-1}$ (for contiguous super-regular)
for various values of $q$.
These are compared to the analytically known values (green curves) of 
$\ProbSR{q}{3}$ and $\ProbCSR{q}{3}$,
and also to the limiting curve $e^{-\lambda}$ (red curves).

For the contiguous case, it can be seen that when $k=10$, 
the experimentally computed probabilities are very close to the limiting
probability $e^{-\lambda}$,
so the convergence appears quite fast.
For the non-contiguous case,
the experimentally computed probabilities for $k=10$ are even closer to the 
conjectured limiting probability $e^{-\lambda}$,
and even the 
experimentally computed probabilities for $k=3$ are very close to the limiting curve,
thus supporting our conjecture.

\clearpage

% -------------------------
\newcommand{\SRplot}{
\begin{tikzpicture}[scale=2.8]

\node[black, scale=1] at (2.5,1.7) {\Large{Super-regular}};

% boxes
\draw[->, ultra thick]  (0,0) -- (5.1,0) node[right=2mm] {$\lambda$};
\draw[->, ultra thick]  (0,0) -- (0,\Inflate * 1.05);

\draw[red, ultra thick] plot[smooth,samples=100,domain=0:5] (\x, {\Inflate * exp(-\x)}); 

\draw[green, ultra thick] plot[smooth,samples=100,domain=4:100] 
  (20 / \x, {\Inflate * (1 -(1/\x))^2 * (1 -(1/\x))^2 * (1 -(1/\x)) * (1 - (2/\x)) * (1 - (3/\x)) * (1 - (9/\x) + (21/\x^2)) }  ); 

\draw[->, dashed]  (1.2, 0.20) -- (1.7, 0.28);  % green for k=3
\draw[->, dashed]  (2.5, 0.7) -- (2.5, 0.18); % red

\node[black, scale=1] at (1.0,0.2) {$k=3$};
\node[black, scale=1] at (2.5,0.8) {$e^{-\lambda}$};

% Data generated by Spencer 08-03-25...
% k = 10  trials = 1000
% lambda      prob      e^{-lambda}  Ken's prob   q
%
% 1/3         0.706     0.717          0.717   554417
\draw[fill=blue](1/3, .717 * \Inflate) circle (0.04cm);
% 1/2         0.629     0.607          0.605   369637
\draw[fill=blue](1/2, .605 * \Inflate) circle (0.04cm);
% 1           0.365     0.368          0.365   184823
\draw[fill=blue](1,   .365 * \Inflate) circle (0.04cm);
% 2           0.125     0.135          0.127    92401
\draw[fill=blue](2,   .127 * \Inflate) circle (0.04cm);
% 3           0.057     0.0498         0.043    61603
\draw[fill=blue](3,   .043 * \Inflate) circle (0.04cm);
%

%\node[black, below, scale=1] at (2, \Inflate * .3) {$e^{-\lambda}$};

\node[black, left, scale=1] at (0,\Inflate) {$1$};

\node[black, below, scale=1] at (0,0) {$0$};
\node[black, below, scale=1] at (0,0) {$0$};
\node[black, below, scale=1] at (1,0) {$1$};
\node[black, below, scale=1] at (2,0) {$2$};
\node[black, below, scale=1] at (3,0) {$3$};
\node[black, below, scale=1] at (4,0) {$4$};
\node[black, below, scale=1] at (5,0) {$5$};
\end{tikzpicture}
}

% ---------------------------------------------------------------------------
\newcommand{\CSRplot}{
\begin{tikzpicture}[scale=2.8]
\node[black, scale=1] at (2.5,1.7) {\Large{Contiguous super-regular}};

% boxes
\draw[->, ultra thick]  (0,0) -- (5.1,0) node[right=2mm] {$\lambda$};
\draw[->, ultra thick]  (0,0) -- (0,\Inflate * 1.05);

\draw[red, ultra thick] plot[smooth,samples=100,domain=0:5] (\x, {\Inflate * exp(-\x)}); 

% k=3, 8/3 = 2.66...  3/8 = 0.375
\draw[green, ultra thick] plot[smooth,samples=100,domain=2:100] 
  (9 / \x, {\Inflate * (1 -(1/\x))^2 * (1 -(1/\x))^2 * (1 -(1/\x)) * (1 - (2/\x)) * (1 - (3/\x)) * (1 - (4/\x) + (5/\x^2)) }  ); 

% Data generated by Ken 08-09-25...
% k = 10  trials = 1000 lambda = k^3/3q = 1000/3q
% lambda      prob      e^{-lambda}      q
%
% 1/3         0.635     0.7273          997
\draw[fill=blue](0.334336, .662 * \Inflate) circle (0.04cm);
% 1/2         0.49      0.6089          673
\draw[fill=blue](0.495295, .556 * \Inflate) circle (0.04cm);
% 1           0.238     0.3955          331
\draw[fill=blue](1.00705,  .332 * \Inflate) circle (0.04cm);
% 2           0.044     0.1438          167
\draw[fill=blue](1.99601,  .099 * \Inflate) circle (0.04cm);
% 3           0.005     0.04747         113
\draw[fill=blue](2.94985,  .043 * \Inflate) circle (0.04cm);

\node[black, left, scale=1] at (0,\Inflate) {$1$};

\draw[->, dashed]  (0.5,0.3) -- (0.75, 0.55);  % green
\draw[->, dashed]  (2.5,0.8) -- (2.5, 0.19); % red

\node[black, scale=1] at (0.3,0.2) {$k=3$};
\node[black, scale=1] at (2.5,0.95) {$e^{-\lambda}$};

\node[black, below, scale=1] at (0,0) {$0$};
\node[black, below, scale=1] at (0,0) {$0$};
\node[black, below, scale=1] at (1,0) {$1$};
\node[black, below, scale=1] at (2,0) {$2$};
\node[black, below, scale=1] at (3,0) {$3$};
\node[black, below, scale=1] at (4,0) {$4$};
\node[black, below, scale=1] at (5,0) {$5$};
\end{tikzpicture}
}
% ---------------------------------------------------------------------------

\renewcommand{\Inflate}{2.0}
\begin{figure}[ht]
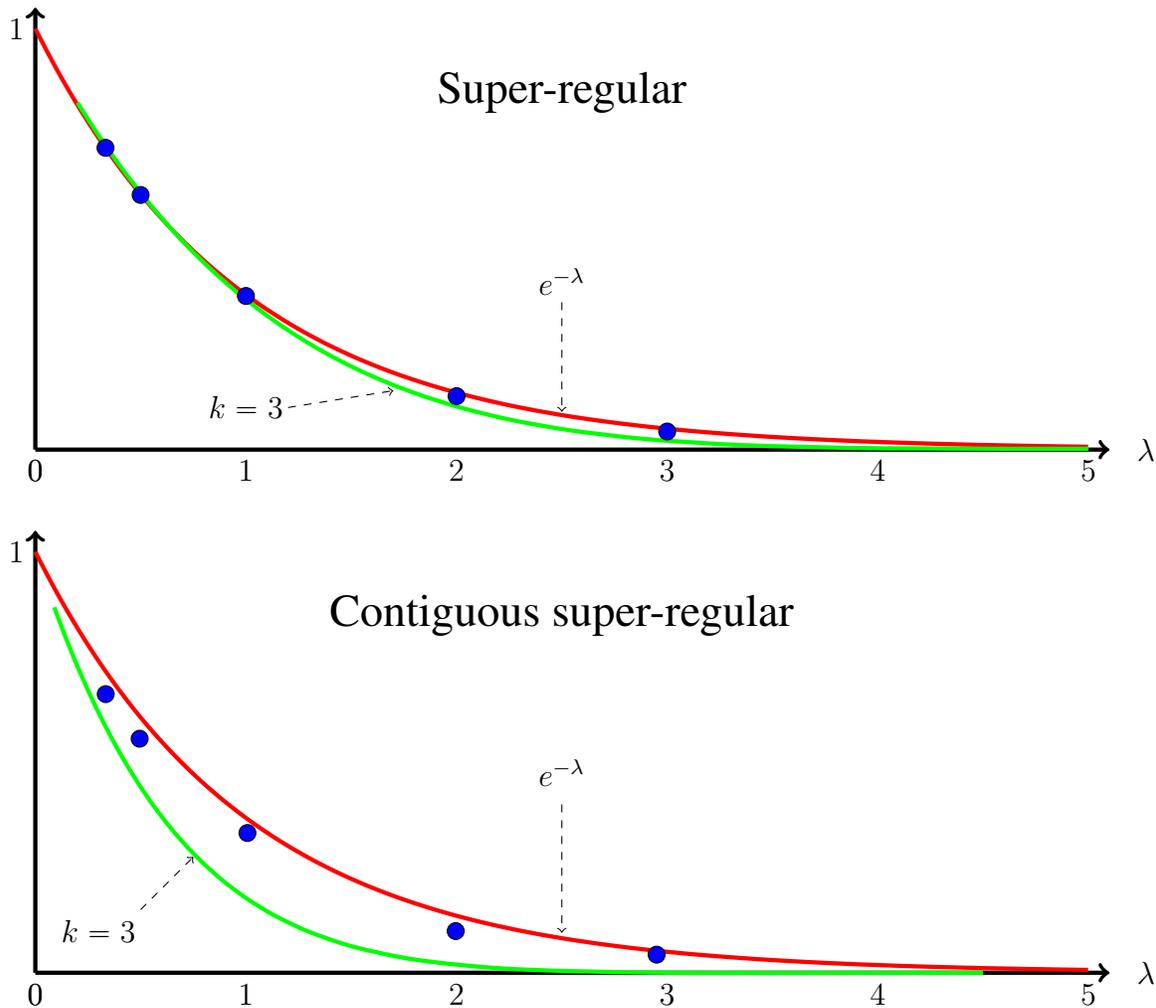


\begin{center}
\SRplot
\end{center}

\begin{center}
\CSRplot
\end{center}

\caption{
Plots of the probability of a random $k\times k$ matrix 
over $\fq$ being super-regular or contiguous super-regular.
The red curves are $e^{-\lambda}$.
The green curves are counts of $3\times 3$ matrices based on 
Theorem~\ref{thm:3x3}(a) for super-regular and (b) for contiguous super-regular.
Each blue dot is the fraction of $1000$ randomly selected $k\times k$ 
matrices over $\fq$ that were 
super-regular (where $\lambda = \frac{1}{q}\binom{2k}{k}$)
or 
contiguous super-regular (where $\lambda = \frac{1}{3} k^3q^{-1}$),
with $k=10$.
}
\label{fig:conjecture}
\end{figure}

% ---------------------------------------------------------------------------
\clearpage

\section{Appendix - Some lemmas}
\label{sec:Appendix}

The lemmas in this section are used
to prove our main results.

We examine the probability that a (uniformly) randomly chosen linear code is MDS.
In principle, 
one could choose a linear code randomly from an exhaustive list of all linear codes,
but this would have impractical complexity.
A number of alternative ways exist,
such as randomly choosing a generator matrix or parity check matrix,
optionally in systematic form.
Lemma~\ref{lem:multiple-defs-of-code} 
demonstrates that these variations are asymptotically equivalent
in terms of the probabilities of the resulting codes being MDS.
First, we give Lemma~\ref{lem:MDS-multiply} 
whose proof follows from basic linear algebra and is omitted.
This lemma is used in the proof of Lemma~\ref{lem:multiple-defs-of-code}.

\begin{lemma} 
Let $A$ be an invertible $k\times k$ matrix over $\fq$.
If $G$ is a $k\times n$ generator matrix for an MDS code over $\fq$, 
then so is $AG$ for the same code.
\label{lem:MDS-multiply}
\end{lemma}

%---------------------------
The following lemma
shows that the probability 
a random $n\times n$ matrix is non-singular is asymptotically one,
no matter how
the field size $q$ and the matrix dimension $n$ tend to infinity.

\begin{lemma}
If $q$ and $k$ are positive integers, then
\begin{align*}
\lim_{q,k\to\infty} 
\prod_{i=1}^k \left( 1 - q^{-i} \right) = 1.
\end{align*}
\label{lem:nonsingular-prob}
\end{lemma}

\begin{proof}
\begin{align}
-\ln \prod_{i=1}^k \left( 1 - q^{-i} \right) 
 &= -\sum_{i=1}^k \ln \left( 1 - q^{-i} \right)
 =  \sum_{i=1}^k \sum_{j=1}^\infty \frac{q^{-ij}}{j} 
 =  \sum_{j=1}^\infty \frac{1}{j} \sum_{i=1}^k q^{-ij}
  = \sum_{j=1}^\infty \frac{1}{j} \cdot  \frac{1-q^{-jk}}{q^j-1} \notag\\
&\le \sum_{j=1}^\infty \frac{1}{q^j-1} 
 \le \sum_{j=1}^\infty \frac{2}{q^j} 
= \frac{2}{q-1}
\to 0\ \ \ \text{as}\ q\to\infty .
\end{align}
Therefore, for any fixed $k$ 
$$\lim_{q \to \infty} \pi (q,k) = 1$$
and also
$$\lim_{q,k \to \infty} \pi (q,k) = 1.$$
\end{proof}

\begin{lemma}
\label{lem:multiple-defs-of-code}
Define the following probabilities:
\begin{itemize}
\setlength\itemsep{-0.15cm}
\item $p_1 = $ probability a random $[n,k]$ linear code over $\fq$ is MDS.
\item $p_2 = $ probability a random $k{\times} n$ matrix over $\fq$ is an MDS generator matrix.
\item $p_3 = $ probability a random $k{\times} n$ matrix over $\fq$ of rank $k$ is an MDS generator matrix.
\item $p_4 = $ probability a random $k{\times} n$ systematic matrix over $\fq$ is an MDS generator matrix.
\item $p_5 = $ probability a random $(n{-}k){\times} n$ matrix over $\fq$ is an MDS parity check matrix.
\item $p_6 = $ probability a random $(n{-}k){\times} n$ matrix over $\fq$ of rank $n{-}k$ is an MDS parity check matrix.
\item $p_7 = $ probability a random $(n{-}k){\times} n$ systematic matrix over $\fq$ is an MDS parity check matrix.
\end{itemize}
Then, as $q,n,k\To\infty$,
if any of the limits of 
$p_1, p_2, p_3, p_4, p_5, p_6, p_7$ 
exist, then they all exist and are equal.
\end{lemma}

\begin{proof}
\ 

\begin{itemize}
\item
Every $k\times n$ matrix over $\fq$ of rank $k$ generates exactly one
$[n,k]$ linear code, namely the row space of the matrix.
In fact,
any $k\times k$ invertible matrix multiplied by the $k\times n$ matrix yields
another $k\times n$ matrix of rank $k$ generating the same code.
So each $[n,k]$ linear code has a collection of the same number of generator matrices,
namely the number of invertible $k\times k$ matrices over $\fq$.
Thus, the probability that a random $[n,k]$ linear code over $\fq$ is MDS
is the same as the probability that a random $k\times n$ matrix over $\fq$ of rank $k$ is an
MDS generator matrix. 
That is, $p_1=p_3$ for all $q$.
A similar argument shows $p_1 = p_6$ for all $q$.

\item
A $k\times n$ matrix $[I_k | A]$ is a systematic generator matrix for an $[n,k]$ linear code
if and only if the $(n-k)\times n$ matrix $[-A^t | I_{n-k}]$ is a systematic parity check matrix 
for the same code.
Thus, $p_4 = p_7$.

\item
% $\lim p_2 = \lim p_3$.
Let $A$ be a random $k\times n$ matrix over $\fq$.
Note that the rank of $A$ equals $k$ 
if and only if the $k$ rows of $A$ are linearly independent.
Then Lemma~\ref{lem:fullColRankProb} 
(taking $d=n$, $m=k$, and $j=0$)
and the fact that $n\ge k$ 
imply
\begin{align*}
P(\text{rank}(A)<k)
&\le 2\cdot q^{-n+k-1} 
\To 0.
\end{align*}
Thus,
\begin{align*}
p_2 &=
P(A\ \text{is MDS})\\
&=        P(A\ \text{is MDS} \mid \text{rank}(A)=k)\underbrace{P(\text{rank}(A)=k)}_{\To 1}\\
&\ \ \  + P(A\ \text{is MDS} \mid \text{rank}(A)<k)\underbrace{P(\text{rank}(A)<k)}_{\To 0}\\
&\therefore p_2 - p_3 \To 0.
\end{align*}
This means $\lim p_2 = \lim p_3$ if either limit exists.
A similar argument for parity check matrices implies $\lim p_5 = \lim p_6$
if either of those two limits exists.

\item
% ---------------------------------------------------------------------------
Let $M_1$ denote the set of all $k\times n$ generator matrices
of MDS codes over $\fq$,
let $M_2$ be the set of all $k\times (n-k)$ super-regular matrices over $\fq$, and
let $M'$ denote the set of all $k\times k$ invertible matrices over $\fq$.

For each $B\in M_2$, 
define $B^* = \{ [A|AB]: A \in M'\}$
and note that its elements are all generator matrices in $M_1$ for the same MDS code
(by Lemma~\ref{lem:MDS-conditions} %implies that $A$ is invertible 
and Lemma~\ref{lem:MDS-multiply}). % implies $[I|A^{-1}B]\in M_2$.
Furthermore, the sets $B^*$ (where $B\in M_2$)
partition $M_1$ into sets each of size $|M'|$, so
$|M_1|=|M_2|\cdot |M'|$.
Therefore,
\begin{align}
p_2 = \frac{|M_1|}{q^{kn}} 
&= \frac{|M_2|}{q^{k(n-k)}} \cdot \frac{|M'|}{q^{k^2}}
= p_4\cdot \frac{|M'|}{q^{k^2}}\notag\\
\frac{p_2}{p_4} &= \frac{|M'|}{q^{k^2}} \To 1 \label{eq:3000}
\end{align}
where 
\eqref{eq:3000} follows 
from Lemma~\ref{lem:nonsingular-prob}.
Thus,
$\lim p_2 = \lim p_4$
if either of these limits exists.
\end{itemize}
\end{proof}

% ---------------------------------------------------------------------------

Lemma~\ref{lem:calc} follows from elementary calculus, and its proof is omitted.
It is used in the proof of Lemma~\ref{lem:convergence}.
%
%----------------------------------
\begin{lemma}\label{lem:calc}\ 
\begin{enumerate}[(a)]
\item For all $z\in(0,\infty)$, we have $\ln z \le z-1$.
\item
For all $\epsilon > 0$ there exists $w \in (0,1)$ such that for all $z \in (w,1)$,
we have $(1+\epsilon) (z-1) \le \ln z$.
\end{enumerate}
\end{lemma}
%----------------------------------

The following lemma is well known
and is used in the proofs of 
Lemma~\ref{lem:two-bad-values-strong2} 
and
Lemma~\ref{lem:Cauchy2}.

\begin{lemma} \label{lem:Delta}
If $\Delta \in [0,1)$ and $s \ge 1$, then
$1-s\Delta \le (1-\Delta)^s$.
\end{lemma}

The following lemma is used in the proofs of
Lemma~\ref{lem:two-bad-values-strong2},
Lemma~\ref{lem:strong-inf},
and
Lemma~\ref{lem:Cauchy2}.

%----------------------------------
\begin{lemma} \label{lem:convergence}
Let $\{a_n\}$ and  $\{b_n\}$ be positive real sequences such that
$b_n \To 0$ as $n \To \infty$,
and let $\Delta \ge 0$.
If\ \  $a_n b_n \To \Delta$,
then $(1-b_n)^{a_n} \To e^{-\Delta}$.
\end{lemma}
\begin{proof}
Let $\epsilon > 0$.
For sufficiently large $n$, 
\begin{align}
-b_n (1+\epsilon) 
&\le \ln(1 - b_n) \label{eq:3}\\ 
&\le -b_n \label{eq:4}\\
e^{-(1+\epsilon) a_n b_n} 
&\le (1-b_n)^{a_n} \label{eq:5}\\
&\le e^{-a_n b_n}  \label{eq:6}
\end{align}
where
\eqref{eq:3} follows from Lemma~\ref{lem:calc}(b);
\eqref{eq:4} follows from Lemma~\ref{lem:calc}(a);
\eqref{eq:5} follows from \eqref{eq:3};
and
\eqref{eq:6} follows from \eqref{eq:4};
and therefore
\begin{align*}
e^{-(1+\epsilon)\Delta} 
\le \liminf_{n \To \infty} (1-b_n)^{a_n} 
\le \limsup_{n \To \infty} (1-b_n)^{a_n}
\le e^{-\Delta}.
\end{align*}
Taking $\epsilon \To 0$ gives
\begin{align*} %\label{Eq:part3_1}
\lim_{n \To \infty} (1-b_n)^{a_n} = e^{-\Delta}.
\end{align*}
\end{proof}

% ---------------------------------------------------------------------------

The following lemma 
counts the number of contiguous square submatrices of an $n\times n$ matrix 
and is used in the proof of 
Lemma~\ref{lem:strong-inf}.

\begin{lemma}
\begin{align*}
\sum_{i=1}^k \sum_{j=1}^k \minij
&= \frac{k^3}{3} + \frac{k^2}{2} + \frac{k}{6} .
\end{align*}
\label{lem:num_connected}
\end{lemma}

\begin{proof}
\begin{align*}
\sum_{i=1}^k \sum_{j=1}^k \minij
%
%   (i < j)   +   (j < i)   +  (i = j)   =  2  (i < j) + (i = j)
&= \sum_{i<j} \minij + \sum_{j<i} \minij + \sum_{i=j} \minij\\
&= \left( 2\sum_{j=1}^{k-1} \sum_{i=j+1}^k j \right) + \sum_{i=1}^k i\\
&= 2\sum_{j=1}^{k-1} j(k-j) + \sum_{j=1}^k j\\
&= k + (2k+1)\sum_{j=1}^{k-1} j - 2\sum_{j=1}^{k-1} j^2\\
&= k + (2k+1)\frac{k(k-1)}{2} - \frac{k(k-1)(2k-1)}{3}\\
&=\frac{k^3}{3} + \frac{k^2}{2} + \frac{k}{6}.
\end{align*}
\end{proof}

%---------------------------------------------------------------------------

The following lemma is used in the proof of Lemma~\ref{lem:two-bad-values-strong2}.

\begin{lemma} \label{lem:uniform}
Suppose $B$ is a matrix with entries in $\fq$ and $x$ is a random vector
with length equal to the width of $B$ and
whose entries are iid uniform on $\fq$.
If $B$ has full row-rank,
then the entries of $Bx$ are iid uniform on $\fq$.
\end{lemma}

\begin{proof}
First note that it suffices to restrict attention to
square matrices,
since any rectangular matrix $B$ that is wider than tall and 
with full-row rank can be extended to an invertible 
square matrix $B'$ by adding more rows to $B$ (e.g. using the Gram-Schmidt method).
The resulting square matrix will satisfy the lemma and the entries of $Bx$ are a subset of
the entries of $B'x$
and therefore will be iid uniform over $\fq$.

So without loss of generality suppose $B$ has dimensions $k\times k$,
is invertible,
and let $y=Bx$.
Let $s,t\in\fq^k$ be such that
$t = B^{-1}s$.
Then,
\begin{align*}
&P(y = s)
 = P(x = B^{-1}s)
 = q^{-k}
\end{align*}
so the vector $y$ is uniform on $\fq^k$,
and therefore each component $y_1, \dots , y_k$ 
is uniform on $\fq$.
Also,
\begin{align*}
P( y_1 = s_1, \dots, y_k = s_k)
&= P(x_1 = t_1, \dots, x_k = t_k)
 = P(x_1 = t_1) \cdots P(x_k = t_k)\\
&= q^{-k}
 = \prod_{i=1}^k P(y_i = s_i)
\end{align*}
so $y_1, \dots, y_k$ are independent.
\end{proof}

The next lemma gives us an inequality used in the proof of
Lemma~\ref{lem:fullColRankProb}.

\begin{lemma} 
\label{lem:inequality}
If $a_1, a_2, \dots \in [0,1]$
and $N$ is a positive integer,
then
$\displaystyle\prod_{i=1}^N (1-a_i) + \displaystyle\sum_{i=1}^N a_i \ge 1$.
\end{lemma}

\begin{proof}
If
$S_N = \displaystyle\prod_{i=1}^N (1-a_i) + \displaystyle\sum_{i=1}^N a_i$,
then $S_1 = 1$ and
\begin{align*}
S_N 
%&= (1-a_N)\prod_{i=1}^{N-1} (1-a_i) + a_N + \sum_{i=1}^{N-1} a_i\\
&= S_{N-1} + a_N \left(1 - \prod_{i=1}^{N-1} (1-a_i) \right)
\ge S_{N-1}.
\end{align*}
\end{proof}

% ---------------------------------------------------------------------------

%
\end{document}